\documentclass{tlp}
\usepackage{aopmath}

\newtheorem{definition}{Definition} 
\newtheorem{example}{Example} 

\def\squareforqed{\hbox{\rlap{$\sqcap$}$\sqcup$}}
\def\qed{\ifmmode\squareforqed\else{\unskip\nobreak\hfil
\penalty50\hskip1em\null\nobreak\hfil\squareforqed
\parfillskip=0pt\finalhyphendemerits=0\endgraf}\fi}

\usepackage{times}
\usepackage{latexsym}
\usepackage{amssymb}
\usepackage{xspace}
\usepackage{url}
\usepackage{graphicx}
\usepackage{epsfig}
\usepackage{alltt}

\usepackage{color}

\newcommand{\dlv}{{\small DLV}\xspace}
\newcommand{\dlvdb}{{\small DLV}$^{DB}$\xspace}

\newcommand{\vars}[1]{\ensuremath{\bar{\textbf{#1}}}}

\newcommand{\relNam}[1]{names(#1)}
\newcommand{\intCon}[1]{constr(#1)}
\newcommand{\indNorm}[1]{$\forall \vars{x}_{\forall} \ [ \ {#1}_1(\vars{x}_1) \rightarrow \exists \vars{x}_{2\exists} \ {#1}_2(\vars{x}_2) \ ]$}
\newcommand{\indMini}[1]{${#1}_1(\vars{x}_1) \rightarrow {#1}_2(\vars{x}_2)$}
\newcommand{\indMicro}[1]{${#1}_1 \rightarrow {#1}_2$}

\newcommand{\pcqa}{\ensuremath{\Pi_{cqa}}\xspace}
\newcommand{\qcqa}{\ensuremath{q_{cqa}}\xspace}
\newcommand{\Qcqa}{\ensuremath{Q_{cqa}}\xspace}

\newcommand{\can}[1]{\ensuremath{#1^{c}}\xspace}

\newcommand{\rep}[1]{\ensuremath{ #1^{r}}\xspace}

\newcommand{\derives}{\ensuremath{\mathtt{\ :\!\!-\ }}}

\newcommand{\Or}{\ensuremath{\vee}}

\newcommand{\naf}{\ensuremath{\mathit{not\ }}}

\newcommand{\CONP}{{\footnotesize \textbf{coNP}}\xspace}
\newcommand{\PiP}[1]{{\footnotesize ${\Pi}_{#1}^{p}$}}
\newcommand{\PTIME}{{\footnotesize \textbf{PTIME}}\xspace}

\newcommand{\G}{\mathcal{G}}
\newcommand{\I}{\mathcal{I}}
\newcommand{\M}{\mathcal{M}}
\newcommand{\D}{\mathcal{D}}
\newcommand{\RET}{\mathcal{D}}
\newcommand{\F}{\mathcal{F}}

\newcommand{\homo}{\textsf{homo}}
\newcommand{\sfsknulls}{\mathcal{N}}
\newcommand{\RETS}{\RET^*}
\newcommand{\vals}[1]{\textsf{dom}(#1)}

\newcommand{\B}{\mathcal{B}}
\renewcommand{\S}{\mathcal{S}}
\newcommand{\tup}[1]{\langle #1\rangle}

\newcommand{\nop}[1]{}

  {\end{list}}

\newcommand{\citecqa}{\cite{Lenzerini02,ArenasBertossiChomicki03,BertossiHunterSchaub05,ChomickiMarcinkowski05}}

\newcommand{\longShort}[2]{#1}

\newcommand{\disclaimer}{\longShort{}{In the following, some of the proofs are not included due to space constraints;
they can be found in the extended version of this paper \cite{tplpOnCoRR}.}}

 \submitted {23 November 2010}
 \accepted{24 January 2011}
\begin{document}

\long\def\comment#1{}

\title[CQA via ASP from different perspectives]{Consistent Query Answering via ASP \\ from Different Perspectives: \\ Theory and Practice }

\author[M. Manna, F. Ricca and G. Terracina]
{MARCO MANNA, FRANCESCO RICCA \and GIORGIO TERRACINA\\
Department of Mathematics, University of Calabria, Italy\\
\email{\{manna,ricca,terracina\}@mat.unical.it}}


\pagerange{\pageref{firstpage}--\pageref{lastpage}}
\volume{\textbf{10} (3):}
\jdate{mm yyyy}
\setcounter{page}{1}
\pubyear{yyyy}

\maketitle

\label{firstpage}

\begin{abstract}
A data integration system provides transparent access to different data sources by suitably combining
their data, and providing the user with a unified view of them, called \emph{global schema}.
However, source data are generally not under the control of the data integration process,
thus integrated data may violate global integrity constraints even in presence of locally-consistent data
sources. In this scenario, it may be anyway interesting to retrieve as much consistent information as possible.
The process of answering user queries under global constraint violations is called
{\em consistent query answering} (CQA).
Several notions of CQA have been proposed, e.g., depending on whether integrated information is
assumed to be {\em sound},  {\em complete}, {\em exact} or a variant of them. 
This paper provides a contribution in this setting: it uniforms solutions coming from different perspectives
under a common ASP-based core, and provides query-driven optimizations designed for isolating and eliminating
inefficiencies of the general approach for computing consistent answers.
Moreover, the paper introduces some new theoretical results enriching existing knowledge on decidability
and complexity of the considered problems.
The effectiveness of the approach is evidenced by experimental results.

\longShort{\noindent To appear in Theory and Practice of Logic Programming (TPLP).}{}

\end{abstract}


\section{Introduction}\label{sec:introduction}

The enormous amount of information dispersed over many data sources, often stored in different
heterogeneous databases, has recently boosted the interest for data integration systems
\cite{Lenzerini02}.
Roughly speaking, a data integration system provides transparent access to different data sources by suitably combining their data, and providing the user with a unified view of them, called \emph{global schema}.
In many cases, the application domain imposes some consistency requirements on integrated data.
For instance, it may be at least desirable to impose some integrity constraints (ICs),
like primary/foreign keys, on the global relations.
It may be the case that data stored at the sources may violate global ICs when integrated, since in general
data sources are not under the control of the data integration process.
The standard approach to this problem basically consists of explicitly modifying the data in order to eliminate
IC violations (data cleaning). However, the explicit repair of data is not always convenient or possible.
Therefore, when answering a user query, the system should be able to ``virtually repair'' relevant data
(in the line of
\citeNP{ArenasBertossiChomicki03,BertossiHunterSchaub05,ChomickiMarcinkowski05}), in order to
provide consistent answers; this task is also called Consistent Query Answering (CQA).

The database community has spent considerable efforts in this area, relevant research results have been
obtained to clarify semantics, decidability, and complexity of \ data-integration under constraints
and, specifically, for CQA. In particular, several notions of CQA have been proposed (see
\citeNP{BertossiHunterSchaub05} for a survey), e.g. depending  on whether the information in the
database is assumed to be {\em sound}, {\em complete} or {\em exact}.
%
%
However, while efficient systems are already available for simple data integration scenarios,
solutions being both scalable and comprehensive have not been implemented yet for CQA, mainly due
to the fact that handling
inconsistencies arising from constraints violation is inherently hard. Moreover, mixing different
kinds of constraints (e.g. denial constraints, and inclusion dependencies) on the same
global database makes, often, the query answering process
undecidable~\cite{AbiteboulHullVianu95,CaliLemboRosatiPODS03}.

This paper provides some contributions in this setting. Specifically, it first starts from different
state-of-the-art semantic perspectives
\cite{ArenasBertossiChomicki03,CaliLemboRosatiPODS03,ChomickiMarcinkowski05}
and revisits them in order to provide a uniform, common core based on Answer Set Programming (ASP)~\cite{GelfondLifschitz88,GelfondLifschitz91}.
Thus, it provides query driven optimizations, in the light of the experience we gained in the
INFOMIX~\cite{LeoneGrecoIanni05} project in order to overcome the limitations observed in
real-world scenarios.
%
%
%
The main contributions of this paper can \mbox{be summarized in:}
\begin{itemize}
\item A theoretical analysis of considered semantics which extends previous results.

\item The definition of a unified framework for CQA based on a purely
declarative, logic based approach which supports the most relevant semantics assumptions on source
data. 
Specifically, the problem of consistent query answering is reduced to cautious
reasoning on (disjunctive) ASP programs with aggregates \cite{fabe-etal-2008-aij} automatically built from both the query
and involved constraints.

\item The definition of an optimization approach designed  to {\em (1)}  ``localize'' and limit
the inefficient part of the computation of consistent answers to small fragments of the input,
{\em (2)} cast down the computational complexity of the repair process \mbox{if possible.}

\item The implementation of the entire framework in a full fledged prototype system.

\item The capability of handling large amounts of data, typical of real-world data integration
scenarios, using as internal query evaluator the \dlvdb
\cite{TerracinaLeoneLioPanetta08}~
system; indeed, \dlvdb allows for
mass-memory database evaluations and distributed data management features.
\end{itemize}  

In order to assess the effectiveness of the proposed approach, we carried out experimental
activities both on a real world scenario and on synthetic data, comparing its behavior on different semantics
and constraints.

The plan of the paper is as follows. Section \ref{sec:dataintegration} formally
introduces the notion of
CQA under different semantics and some new theoretical results on decidability
and complexity for this problem. Section \ref{sec:cqa} first introduces a unified (general) solution
to handle CQA via ASP, and then
presents some optimizations. Section \ref{sec:experiments} describes the benchmark framework we
adopted in the tests and discusses on obtained results.
Finally, Section \ref{sec:conlusions} compares related work and
draws some conclusive considerations.

\section{Data Integration Framework}\label{sec:dataintegration}

In this paper we exploit the data integration setting to point out motivations and
challenges underlying CQA. However, as it will be clarified in the following, techniques
and results provided in the paper hold also for a single database setting.
We next formally describe the adopted data integration framework.

The following notation will be used throughout the paper.
We always denote by $\Gamma$ a countably infinite domain of totally ordered values;
by $t$ a tuple of values from $\Gamma$;
by $X$ a variable;
by $\vars{x}$ a sequence $X_1,\ldots,X_n$ of (not necessarily distinct) variables,
and by $|\vars{x}| = n$ its length.
%
%
Let $\vars{x}, \vars{x}'$ be two sequences of variables,
we denote by $\vars{x} -  \vars{x}'$ the  sequence
obtained from $\vars{x}$ by discarding a variable if it appears in $\vars{x}'$.
Whenever all the variables of sequence $\vars{x}$
appear in another sequence $\vars{x}'$, we simply write $\vars{x} \leq \vars{x}'$.
Given a sequence $\vars{x}$ and a set $\pi \subseteq \{1,\ldots,|\vars{x}|\}$,
we denote by $\vars{x}^{\pi}$ the sequence obtained from $\vars{x}$
by discarding a variable if its position is not in $\pi$.
(Similarly, given a tuple $t$ and a set $\pi \subseteq \{1,\ldots,|t|\}$,
we denote by $t^{\pi}$ the tuple obtained from $t$
by discarding a value if its position is not in $\pi$.)
Moreover, we denote, by $\sigma(\vars{x})$ a conjunction of
comparison atoms of the form $X \odot X'$, where $\odot \in \{\leq,\geq,<,>,\neq\}$,
and by $\ominus$, the symmetric difference operator between \mbox{two sets.}

\longShort{\medskip}{}


A relational \emph{database schema} is a pair $\mathcal{R} = \langle \relNam{\mathcal{R}},
\intCon{\mathcal{R}}\rangle$ where $\relNam{\mathcal{R}}$ and $\intCon{\mathcal{R}}$ are the relation names
and the integrity constraints (ICs) of $\mathcal{R}$, respectively. The arity of a given relation
$r \in \relNam{\mathcal{R}}$ is denoted by $arity(r)$. 
A \emph{database} (instance) for $\mathcal{R}$ is any set of facts \cite{AbiteboulHullVianu95} of the form:
\[
    \F = \{r(t) : r \in \relNam{\mathcal{R}} \ \wedge \ t \ \textrm{is a tuple from} \ \Gamma \ \wedge \ |t| = arity(r) \}
\]
\noindent In the following, we adopt the \emph{unique name assumption}, and
$\vals{\F}$ denotes the subset of $\Gamma$ containing all the values appearing in the facts of $\F$.

Let $r_1, \ldots, r_m \in \relNam{\mathcal{R}}$, the set $\intCon{\mathcal{R}}$ contains ICs of
the form:
\begin{enumerate}
  \item $\forall \vars{x}_1, \ldots, \vars{x}_m \ \ \neg [ \ r_1(\vars{x}_1) \wedge \ldots \wedge r_m(\vars{x}_m) \wedge \sigma(\vars{x}_1, \ldots, \vars{x}_m) \
  ]$ (\emph{denial constraints} -- DCs)



  \item \indNorm{r} (\emph{inclusion dependencies} -- INDs);
\end{enumerate}
where 
$arity(r_i) = |\vars{x}_i|$, for each $i$ in [1..$m$].
In particular, for INDs we require that all the variables within an $\vars{x}_i$ ($1 \leq i \leq 2$)
are distinct,
$\vars{x}_{\forall} \leq \vars{x}_1$, $\vars{x}_{\forall} \leq \vars{x}_2$,
and $\vars{x}_{2\exists} = \vars{x}_2 - \vars{x}_{\forall}$.
Note that, if $|\vars{x}_{2\exists}| = 0$, then $\vars{x}_{\forall} = \vars{x}_2 \leq \vars{x}_1$.
%
%
In the case we are only interested in emphasizing the relation names involved in an IND, we simply write
\indMini{r} or \indMicro{r}.
%
A database $\F$ is said to be \emph{consistent} w.r.t. $\mathcal{R}$ if all ICs are satisfied.
A \emph{conjunctive query} $cq(\vars{x})$ over $\mathcal{R}$ is a formula of the form
%
\[
    \exists \vars{x}_{1\exists}, \ldots, \vars{x}_{m\exists} \ \ r_1(\vars{x}_1) \wedge \ldots \wedge r_m(\vars{x}_m) \wedge \sigma(\vars{x}_1 \ldots, \vars{x}_m)
\]
where $\vars{x}_{i\exists} \leq \vars{x}_i$ for each  $i$ in [1..$m$],
$\vars{w} = \vars{x}_1 - \vars{x}_{1\exists}, \ldots, \vars{x}_m - \vars{x}_{m\exists}$
are the \emph{free variables} of $q$, and $\vars{x}$ contains only and all
the variables of $\vars{w}$ (with no duplicates, and possibly in different order).
%
%
%
A \emph{union of conjunctive queries} $q(\vars{x})$ is a formula of the form
$cq_1(\vars{x}) \vee \ldots \vee cq_n(\vars{x})$.
In the following, for simplicity, the term query refers to a union of conjunctive queries, if not differently specified.
Given a database $\F$ for $\mathcal{R}$, and a query $q(\vars{x})$, the \emph{answer} to $q$ is the set of $n$-tuples of values
$ans(q,\F) = \{t: \F \models q(t)\}$.


\subsection{The Data Integration Model}

A {\em data integration system} is formalized \cite{Lenzerini02} as a triple  $\I =
\tup{\G,\S,\M}$ where
\vspace{-1.5mm}
\begin{enumerate}
  \item[$\centerdot$] $\G$ is the {\em global schema}.
  A \emph{global database} for $\I$ is any database for $\G$;

  \item[$\centerdot$] $\S$ is the {\em source schema}. A \emph{source database} for $\I$ is any database consistent w.r.t. $\S$;

  \item[$\centerdot$] $\M$ is the \emph{global-as-view} (GAV) {\em mapping}, that associates each element $g$ in $\relNam{\G}$ with a union of conjunctive queries over $\S$.

\end{enumerate}
%
%
%
%
\noindent Let $\F$ be a source database for $\I$. The \emph{retrieved global database} is
\[
    ret(\I,\F) = \{g(t): g \in \relNam{\G} \ \wedge \ t \in ans(q,\F) \ \wedge \ q \in \M(g)\}
\]
for $\G$ satisfying the mapping. Note that, when source data are combined in a unified schema
with its own ICs, the retrieved global database might be inconsistent.

In the following, when it is clear from the context, we use simply the symbol $\RET$ to denote
the retrieved global database $ret(\I,\F)$. In fact, all results provided in the paper hold
for any database $\D$ complying with some schema $\G$ but possibly inconsistent w.r.t. the constraints
of $\G$.


\begin{example}\label{ex-init}
Consider a bank association that desires to unify the databases of two branches. The first
(source) database models managers by using a relation $man(code,name)$ and employees by a
relation $emp(code,name)$, where $code$ is a primary key for both tables. The second database
stores the same data in a relation $employee(code,name,role)$. Suppose that the data have to be
integrated under a global schema with two relations $m(code)$ and $e(code,name)$, where the
global ICs are:
\begin{itemize}

  \item $\forall X_1, X_2, X_3 \ \ \neg [e(X_1,X_2) \wedge e(X_1,X_3) \wedge X_2 \neq X_3]$ namely, $code$ is the key of $e$;

  \item $\forall X_1  [ m(X_1) \rightarrow \exists X_2 \ e(X_1,X_2)]$ i.e., an IND imposing that each manager code must be an employee code as well.
\end{itemize}
The mapping is defined by the following Datalog rules (as usual, see \citeNP{AbiteboulHullVianu95}):
\begin{enumerate}
  \item[] $e(X_c,X_n) \derives emp(X_c,X_n).$ \quad \quad \quad \quad $m(X_c) \derives man(X_c,\_).$

  \item[] $e(X_c,X_n) \derives employee(X_c,X_n,\_).$ \quad $m(X_c) \derives employee(X_c,\_,`manager').$
\end{enumerate}
Assume that, $emp$ stores tuples \emph{(`e1',`john')}, \emph{(`e2',`mary')}, \emph{(`e3',`willy')},
$man$ stores \emph{(`e1',`john')}, and $employee$ stores \emph{(`e1',`ann',`manager')}, \emph{(`e2',`mary',`manager')},
\emph{(`e3', `rose',`emp')}. It is easy to verify that, although the source databases are consistent w.r.t. local constraints,
the global database, obtained by evaluating the mapping,
violates the key constraint on $e$ as both \emph{john} and \emph{ann} have the same code \emph{e1}, and both \emph{willy}
and \emph{rose} have the same code \emph{e3} in table $e$. \qed
\end{example}

\subsection{Consistent Query Answering under different semantics}
\label{sub:semantics}

In case a database $\RET$ violates ICs, one can still be interested in querying the
``consistent'' information originating from $\F$. One possibility is to ``repair'' $\RET$
(by inserting or deleting tuples) in such a way that all the ICs are satisfied. But there are
several ways to ``repair'' $\RET$. As an example, in order to satisfy an IND of the form
\indMicro{r}
one might either remove violating tuples from $r_1$ or insert new tuples in $r_2$. Moreover, the
repairing strategy  depends on the particular semantic assumption made on the data integration
system. Semantic assumptions may range from (strict) soundness to (strict)
completeness. Roughly speaking, completeness complies with the \emph{closed world assumption}
where missing facts are assumed to be false; on the contrary, soundness complies with the {\em
open world assumption} where $\RET$ may be incomplete.
We next define consistent query
answering under some relevant semantics, namely \emph{loosely-exact, loosely-sound, CM-complete}
\cite{ArenasBertossiChomicki03,CaliLemboRosatiPODS03,ChomickiMarcinkowski05}.
More formally, let $\Sigma$ denote a \emph{semantics}, and $\RET$ a possibly inconsistent database for $\G$,
a database $\B$ is said to be a $\Sigma$-\emph{repair} for
$\RET$ if it is consistent w.r.t. $\G$ and one of the following conditions holds:
\begin{enumerate}


  \item $\Sigma = \emph{CM-complete}$, $\B \subseteq \RET$, and $\nexists$ $\B' \subseteq \RET$ such that
  $\B'$ is consistent and $\B' \supset \B$; 


  \item $\Sigma = \emph{loosely-sound}$ and $\nexists$ $\B'$ such that
  $\B'$ is consistent and $\B' \cap \RET \supset \B \cap \RET$;


  \item $\Sigma = \emph{loosely-exact}$, and $\nexists$ $\B'$
  such that $\B'$ is consistent and  $\B' \ominus \RET \subset \B \ominus \RET$.
%
%
\end{enumerate}


The \emph{CM-complete} semantics
allows a minimal number of deletions in each repair to avoid empty repairs, if possible, but does
not allow insertions. The \emph{loosely-sound} semantics allows insertions
and a minimal amount of deletions. Finally, the \emph{loosely-exact} semantics
allows both insertions and deletions by minimization of the symmetric difference between
$\RET$ and the repairs.

\begin{definition}\label{cqa-def}
Let $\RET$ be a database for a schema $\G$, and $\Sigma$ be a semantics.
The \emph{consistent answer} to a query $q$ w.r.t. $\RET$, is the set
$
ans_{\Sigma}(q,\G,\RET) = \{t: t \in ans(q,\B) \ \ \emph{for each} \ \Sigma \emph{-repair} \ \B\ \emph{for}\ \RET \}
$
\emph{Consistent Query Answering} (CQA) is the problem of computing $ans_{\Sigma}(q,\G,\RET)$. \qed
\end{definition}

Observe that other semantics have been considered in the literature, like \emph{sound}, \emph{complete},
\emph{exact}, \emph{loosely-complete}, etc. \cite{CaliLemboRosatiPODS03}; however, some of them
are trivial for CQA; as an example,
in the \emph{exact} semantics  CQA makes sense only if the retrieved
database is already consistent with the global constraints, whereas in the \emph{complete} and
\emph{loosely-complete} semantics CQA will always return a void answer.
Note that, the semantics considered in this paper address a wide significant range of ways to repair the retrieved database
which are also relevant for CQA.

\begin{example}\label{ex-repairs1}
By following Example \ref{ex-init}, the retrieved global database admits exactly the following repairs under the \emph{CM-complete} semantics:
\longShort{\begin{tabbing}
$\ \ \ \ \ \B_1 = \{e\emph{(`e2',`mary')}, \ e\emph{(`e1',`john')}, \ e\emph{(`e3',`willy')}, \ m\emph{(`e1')}, \ m\emph{(`e2')}\}$\\
$\ \ \ \ \ \B_2 = \{e\emph{(`e2',`mary')}, \ e\emph{(`e1',`john')}, \ e\emph{(`e3',`rose')}, \ m\emph{(`e1')}, \ m\emph{(`e2')}\}$\\
$\ \ \ \ \ \B_3 = \{e\emph{(`e2',`mary')}, \ e\emph{(`e1',`ann')}, \ e\emph{(`e3',`willy')}, \ m\emph{(`e1')}, \ m\emph{(`e2')}\}$\\
$\ \ \ \ \ \B_4 = \{e\emph{(`e2',`mary')}, \ e\emph{(`e1',`ann')}, \ e\emph{(`e3',`rose')}, \ m\emph{(`e1')}, \ m\emph{(`e2')}\}$
\end{tabbing}}
{$\B_1 = \{e\emph{(`e2',`mary')}, \ e\emph{(`e1',`john')}, \ e\emph{(`e3',`willy')}$, $m\emph{(`e1')}, \ m\emph{(`e2')}\}$;
$\B_2 = \{e\emph{(`e2',`mary')}, \ e\emph{(`e1',`john')}, \ e\emph{(`e3',`rose')},$ $m\emph{(`e1')}$, $m\emph{(`e2')}\}$;
$\B_3 = \{e\emph{(`e2',`mary')}, \ e\emph{(`e1',`ann')}, \ e\emph{(`e3',`willy')}, \ m\emph{(`e1')}, \ m\emph{(`e2')}\}$;
$\B_4$ $=$ $\{e\emph{(`e2',`mary')}$, $e\emph{(`e1',`ann')}$, $e\emph{(`e3',`rose')}$, $m\emph{(`e1')}, \ m\emph{(`e2')}\}$.}
\noindent Query $m(X)$ asking for the list of manager codes has then both \emph{e1} and
\emph{e2} as consistent answers, whereas the query $e(X,Y)$ asking for the list of employees has
only $e\emph{(`e2',`mary')}$ as consistent answer ($e$ is the only tuple in each \emph{CM-complete} repair). \qed
\end{example}

\subsection{Restricted Classes of Integrity Constraints}

The problem of computing CQA, under general combinations of ICs, is undecidable ~\cite{AbiteboulHullVianu95}.
However, restrictions on ICs to retain decidability and identify tractable cases can be imposed.

\begin{definition}
Let $r$ be a relation name of arity $n$,
and $\pi$ be a set of $m \leq n$ indices from $I = \{1,\ldots,n\}$.
A \emph{key dependency} (KD) for $r$ consists of a set of $n-m$ DCs,
exactly one for each index $i \in I - \pi$, of the form
$
    \forall \vars{x}_1, \vars{x}_2 \ \ \neg (r(\vars{x}_1) \wedge r(\vars{x}_2)  \wedge \ \vars{x}_1^i \neq \vars{x}_2^i)
$
where no variable occurs twice in each $\vars{x}_i$ ($1 \leq i \leq 2$),
$|\vars{x}_1| = |\vars{x}_2| = n$,
the sequence $\vars{x}_1^{\pi}$ exactly coincides with $\vars{x}_2^{\pi}$,
and $\vars{x}_1^{j}$ is distinct from $\vars{x}_2^{j}$ for each $j \in I-\pi$.
The set $\pi$ is called the \emph{primary-key} of $r$ and is denoted by $key(r)$.
We assume that at most one KD is specified for each relation \cite{CaliLemboRosatiPODS03}.
Finally, for each relation name $r'$ such that no DC is explicitly specified for, we say, without loss of generality,
that $key(r') = \{1,\ldots,arity(r')\}$. \qed
\end{definition}

\begin{definition}
Given an inclusion dependency $d$ of the form \indNorm{r},
we denote by $\pi_L^d \subseteq \{1,\ldots,arity(r_1)\}$
and $\pi_R^d \subseteq \{1,\ldots,arity(r_2)\}$ the two sets of indices
induced by the positions of the variables $\vars{x}_{\forall}$ in $\vars{x}_1$ and $\vars{x}_2$,
respectively.
More formally, $\pi_L^d = \{i : \vars{x}_1^i$ \emph{is universally quantified in} $d\}$
and $\pi_R^d = \{i : \vars{x}_2^i$ \emph{is universally quantified in} $d\}$. \qed
%
\end{definition}

For example, let $d$ denote the IND $\forall X_1,X_2 \ [ \ r_1(X_1,X_3,X_2) \rightarrow \exists X_4 \ r_2(X_4,X_2,X_1) \ ]$.
We have that $\pi_L^d = \{1,3\}$ and $\pi_R^d = \{2,3\}$.

\begin{definition}
\longShort{An IND $d$ is said to be
\begin{itemize}
  \item a \emph{foreign key} (FK) if $\pi_R^d = key(r_2)$ \cite{AbiteboulHullVianu95};

  \item a \emph{foreign superkey} (FSK) if $\pi_R^d \supseteq key(r_2)$ \cite{LeveneVincent00};

  \item \emph{non-key-conflicting} (NKC) if $\pi_R^d \not\supset key(r_2)$ \cite{CaliLemboRosatiPODS03}. \qed
\end{itemize}}
{An IND $d$ is said to be
\emph{(i)} a \emph{foreign key} (FK) if $\pi_R^d = key(r_2)$ \cite{AbiteboulHullVianu95};
\emph{(ii)} a \emph{foreign superkey} (FSK) if $\pi_R^d \supseteq key(r_2)$ \cite{LeveneVincent00};
\emph{(iii)} \emph{non-key-conflicting} (NKC) if $\pi_R^d \not\supset key(r_2)$ \cite{CaliLemboRosatiPODS03}.
(Note that each FK is an FSK.) \qed}
\end{definition}


\begin{definition}
An FSK $d$ of the form \indMicro{r} is said to be \emph{safe}
(SFSK) if $\pi_L^d \subseteq key(r_1)$. In particular, if $d$ is a  \emph{safe} FK
we call it an SFK.\qed
\end{definition}

For example, let $d$ denote the FSK $\forall X_1,X_2 \ [ \ r_1(X_1,X_3,X_2) \rightarrow \exists X_4 \ r_2(X_4,X_2,X_1) \ ]$
where $key(r_2) = \{3\}$. Thus, if $key(r_1)=\{1,3\}$, $d$ is SFSK, whereas if $key(r_1)=\{1,2\}$,
$d$ is not SFSK.


\longShort{\medskip}{}

Table \ref{tableResults} summarizes known and new results about computability and complexity of CQA
under relevant classes of ICs and the three semantic assumptions considered in this paper.
In particular, given a query $q$ (without comparison atoms if $\Sigma \in \{loosely$-$sound, loosely$-$exact\}$),
we refer to the decision problem of establishing whether a tuple
from $\vals{\RET}$ belongs to $ans_{\Sigma}(q,\G,\RET)$ or not.
Note that, \citeN{ChomickiMarcinkowski05} have proved
computability and complexity of CQA for the \emph{CM-complete} semantics
in case of conjunctive queries with comparison predicates.
However, since in such a setting there is a finite number of
repairs each of finite size, then their results straightforwardly hold
for union of conjunctive queries as well.
New decidability and complexity results for CQA under KDs and SFSKs only, with
 $\Sigma \in \{\emph{loosely-sound}, \ \emph{loosely-exact}\}$
are proved in Section \ref{sec:newResults}.



\begin{table}[t]
\caption{Data Complexity of CQA (distinguishing between cyclic/acyclic INDs)}
\label{tableResults}
\begin{minipage}{12cm}
\begin{tabular}{c|c|ccc}
  \hline\hline

  \textbf{DCs} \ & \textbf{INDs} & \textbf{\emph{loosely-sound}} & \textbf{\emph{loosely-exact}} & \textbf{\emph{CM-complete}} \\

  \hline\hline

  \emph{no} & \emph{any} & in \PTIME $^{(1)}$ & in \PTIME $^{(1)}$ & in \PTIME $^{(2)}$ \\

  \hline

  \emph{KD} & \emph{no} & \CONP-c $^{(1)}$ & \CONP-c $^{(1)}$ & \CONP-c $^{(2)}$ \\

  \hline





  \emph{KD} & \emph{NKC} & \CONP-c $^{(1)}$ & \PiP{2}-c $^{(1)}$ & in \PiP{2} $^{(2)}$ $/$ in \CONP $^{(2)}$ \\

  \hline

  \emph{KD} & \emph{SFSK} \ \ & in \PiP{2} $^{(3)}$ & in \PiP{2} $^{(3)}$ & in \PiP{2} $^{(2)}$ $/$ in \CONP $^{(2)}$ \\

  \hline

  \emph{KD} \ & \emph{any} & undec. $^{(1)}$ & undec. $^{(1)}$ & in \PiP{2} $^{(2)}$ $/$ in \CONP $^{(2)}$  \\

  \hline

  \emph{any} \ & \emph{any} & undec. $^{(4)}$ & undec. $^{(4)}$ & \PiP{2}-c $^{(2)}$ $/$ \CONP-c $^{(2)}$  \\

  \hline\hline
\end{tabular}
\footnotesize
$^{(1)}$ \citeNP{CaliLemboRosatiPODS03}; \ $^{(2)}$ \citeNP{ChomickiMarcinkowski05}; \
$^{(3)}$ Section \ref{sec:newResults}; \ $^{(4)}$ \citeNP{AbiteboulHullVianu95};
\normalsize
\end{minipage}
\end{table}

\subsection{Loosely-exact and Loosely-sound semantics under KD and SFSK}\label{sec:newResults}

In this section we provide new decidability and complexity results for CQA under both the loosely-exact
and the loosely-sound semantics with KDs and SFSKs.
\longShort{In the rest of the section we always denote by:
\begin{itemize}
   \item $\G$, a schema containing KDs and SFSKs only;
   \item $\RET$, a possibly inconsistent database for $\G$;
   \item $q$, a union of conjunctive queries without comparison atoms.
   \item $\Sigma \in  \{\emph{loosely-exact}, \ \emph{loosely-sound}\}$.
 \end{itemize}
}
{In the rest of the section we always denote
by $\G$, a schema containing KDs and SFSKs only;
by $\RET$, a possibly inconsistent database for $\G$;
by $q$, a union of conjunctive queries without comparison atoms;
and by $\Sigma \in  \{\emph{loosely-exact}, \ \emph{loosely-sound}\}$.
}

\disclaimer

We first show that, in the aforementioned hypothesis, the size of each repair is finite.

\longShort{\begin{definition}\label{def:slicedB}
Let $\B$ be a $\Sigma$-repair for $\RET$ and $i \geq 0$ be a natural number.
We inductively define the sets $\B^i$ as follows:
\begin{enumerate}
  \item If $i=0$, then $\B^0 = \B \cap \RET$.

  \item If $i > 0$, then $\B^i \subseteq \B - (\B^0 \cup \ldots \cup \B^{i-1})$
    is arbitrarily chosen in such a way that
    its facts are necessary and sufficient
    for satisfying all the INDs in $\intCon{\G}$ that are violated in $\B^0 \cup \ldots \cup \B^{i-1}$.
\end{enumerate}
Observe that $\B = \bigcup_{i \geq 0} \B^i$ and that
$\B^i \cap \B^j = \emptyset$ for each $j \neq i$. \qed
\end{definition}}
{\begin{definition}\label{def:slicedB}
Let $\B$ be a $\Sigma$-repair for $\RET$ and $i \geq 0$ be a natural number.
We inductively define the sets $\B^i$ as follows: (1) If $i=0$, then $\B^0 = \B \cap \RET$.
(2) If $i > 0$, then $\B^i \subseteq \B - (\B^0 \cup \ldots \cup \B^{i-1})$
    is arbitrarily chosen in such a way that
    its facts are necessary and sufficient
    for satisfying all the INDs in $\intCon{\G}$ that are violated in $\B^0 \cup \ldots \cup \B^{i-1}$.
Observe that $\B = \bigcup_{i \geq 0} \B^i$ and that
$\B^i \cap \B^j = \emptyset$ for each $j \neq i$. \qed
\end{definition}}

\longShort{\begin{lemma}\label{lem:SFSKrepIsFinite}
Let $\B$ be a $\Sigma$-repair for $\RET$, then
\begin{enumerate}
  \item The key of each fact in $\B$ only contains values from $\vals{\RET}$.

  \item $|\B|$ is finite.
\end{enumerate}
\end{lemma}}
{\begin{lemma}\label{lem:SFSKrepIsFinite}
Let $\B$ be a $\Sigma$-repair for $\RET$, then
(1) The key of each fact in $\B$ only contains values from $\vals{\RET}$;
(2) $|\B|$ is finite.
\end{lemma}}

\begin{proof}
(1) Let $i > 0$ be a natural number.
Let $r_i(t_i)$ be a fact in $\B^i$ such that there is an index $j \in key(r_i)$
for which $t_i^j \not\in \vals{\B^0}$.
Let $r_{i-1}(t_{i-1})$ be one of the facts in $\B^{i-1}$
that forces the presence of $r_i(t_i)$ in $\B^i$
for satisfying some IND, say $d$.
(Note that, by Definition \ref{def:slicedB}, there must be at least one of such a fact because
$\B^i$ would otherwise violate condition 2, since $r_i(t_i)$ would be unnecessary.)
Moreover, since $d$ is a safe FSK, then there must exist an index $k \in key(r_{i-1})$ such that
 $t_{i}^j = t_{i-1}^k$.
Thus, $r_{i-1}(t_{i-1})$ contains a value being not in $\vals{\B^0}$ inside its key as well as $r_i(t_i)$.
Since $i$ has been chosen arbitrarily, then value $t_i^j$ has to be part of a fact of $\B^0$, which is
clearly a contradiction.

\noindent (2)  Since, the key of each fact in $\B$ can only contain values from $\vals{\B^0}$, and
$|\vals{\B^0}| \leq |\B^0| \cdot \alpha$
where $\alpha = \textsf{max}\{arity(g) : g \in \relNam{\G}\}$,
then $|\B| \leq |\relNam{\G}|\cdot|\vals{\B^0}|^{\alpha} \leq |\relNam{\G}|\cdot(\alpha \cdot |\B^0| )^{\alpha}
\leq |\relNam{\G}|\cdot(\alpha \cdot |\RET| )^{\alpha}$.
\end{proof}

We next characterize representative databases for $\Sigma$-repairs.

\begin{definition}\label{def:homoB}
Let $\B$ be a $\Sigma$-repair for $\RET$.
We denote by $\homo(\B)$ the (possibly infinite) set of databases defined in such a way that
$\B' \in \homo(\B)$ if and only if:
\begin{itemize}
  \item $\B'$ can be obtained from $\B$ by replacing
    each value (if any) that is not in $\vals{\RET}$ with a
    value from $\Gamma - \vals{\RET}$; and

  \item none of the values in $\Gamma - \vals{\RET}$ 
     occurs twice in $\B'$.
\end{itemize}
Finally, we denote by $h_{\B,\B'}: \vals{\B'} \rightarrow \vals{\B}$ the  function (homomorphism)
associating values in $\vals{\B'}$ with values in $\vals{\B}$, where
$h_{\B,\B'}(\alpha) = \alpha$, for each $\alpha \in \vals{\RET} \cap \vals{\B'}$. \qed
\end{definition}

\longShort{\medskip}{}

Note that, since (by Lemma \ref{lem:SFSKrepIsFinite}) the key of each fact in $\B$ only contains values from $\vals{\RET}$,
then $|\B'| = |\B|$ holds.

\longShort{\medskip}{}

For example, if $\B = \{p(1,\varepsilon_1,\varepsilon_2), q(2,\varepsilon_2,\varepsilon_1)\}$ with
$\vals{\RET} = \{1,2\}$ and $key(p) = key(q) = \{1\}$, then all of the following databases are
in $\homo(\B)$:
$\{p(1,\varepsilon_1,\varepsilon_3), q(2,\varepsilon_2,\varepsilon_4)\}$,
$\{p(1,\varepsilon_4,\varepsilon_2), q(2,\varepsilon_3,\varepsilon_1)\}$ and
$\{p(1,\varepsilon_5,\varepsilon_6), q(2,\varepsilon_7,\varepsilon_8)\}$.

\longShort{\medskip}{}

\begin{lemma}\label{lemma:homoRep}
If $\B$ is a $\Sigma$-repair for $\RET$, then
each $\B' \in \homo(\B)$ also is.
\end{lemma}

\longShort{\begin{proof}
Let $\B' \in \homo(\B)$. First of all, we prove that $\B'$ is consistent w.r.t. $\G$.
In particular,
since the key of each fact in $\B$ only contains values from $\vals{\RET}$ (by Lemma \ref{lem:SFSKrepIsFinite}),
then $\B'$ cannot violate any KD (by Definition \ref{def:homoB});
Moreover,
since each IND has to be satisfied through values of a key (by definition of safe FSKs), and
since the key of each fact in $\B$ only contains values from $\vals{\RET}$ (by Lemma \ref{lem:SFSKrepIsFinite}),
then $\B'$ cannot violate any IND (by Definition \ref{def:homoB});

We now prove that $\B'$ is a repair, first for the loosely-sound semantics and then for the
loosely-exact semantics.

\longShort{\medskip}{}

\noindent \textbf{[loosely-sound]} If $\Sigma = \emph{loosely-sound}$, then observe that $\B' \cap \RET = \B \cap \RET$, by definition of $\homo(\B)$.
Thus, if $\B'$ was consistent but not a loosely-sound repair there would exist a loosely-sound repair $\B''$ such that
$\B'' \cap \RET \supset \B' \cap \RET = \B \cap \RET$. Contradiction.

\longShort{\medskip}{}

\noindent \textbf{[loosely-exact]} If $\Sigma = \emph{loosely-exact}$, then assume that
$\B$ is a loosely-exact repair but $\B'$ (although consistent w.r.t. $\G$) is not.
By definition, there must be a loosely-exact repair $\B''$ such that
$\B'' \ominus \RET \subset \B' \ominus \RET$. In particular,
we distinguish three cases:
\begin{enumerate}
\item[] \begin{enumerate}
  \item[(1)] $\B'' - \RET = \B' - \RET$  and $\RET - \B''  \subset \RET - \B'$

  \item[(2)] $\B'' - \RET \subset \B' - \RET$  and $\RET - \B'' = \RET - \B'$

  \item[(3)] $\B'' - \RET \subset \B' - \RET$ and $\RET - \B''  \subset \RET - \B'$
\end{enumerate}
\end{enumerate}

\textsc{Case 1}: Since, by Definition \ref{def:homoB}, for each fact in $\B$ there is a fact in $\B'$ with the same key,
if we could add the facts in $\B'' - \B'$ to $\B'$ without
violating any KD, then such facts could also be added to $\B$ without violating any KD.
Moreover, if we could add to $\B'$ the facts in $\B'' - \B'$ without
violating any IND, then such facts could be also added to $\B$ preserving consistency.
This follows by the definition of safe FSKs (because each IND has to be satisfied
through values of a key), by Lemma \ref{lem:SFSKrepIsFinite} (because
the key of each fact in a loosely-exact repair only contains values from $\vals{\RET}$) and,
by Definition \ref{def:homoB} (because for each fact in $\B'$ there is a fact in $\B$ with the same key and
with the same values from $\vals{\RET}$).
Consequently, we could add all the facts in $\B'' - \B'$ to $\B$
preserving consistency. But this is not possible since $\B$ is a loosely-exact repair.

\textsc{Case 2}: Since in $\B'$ we have unnecessary facts (those in $\B' - \B''$) or equivalently
the facts in $\B''$ do not violate any IND, then the corresponding facts in $\B$ do not violate
any IND by Lemma \ref{lem:SFSKrepIsFinite} and by Definition \ref{def:homoB}.
Consequently, if each fact $f \in \B$, such that there is a fact $f' \in \B' - \B''$ that is homomorphic to $f$, was removed from $\B$,
then we would obtain a database preserving consistency and with a smaller symmetric difference than $\B$.
But this is not possible since $\B$ is a loosely-exact repair.

\textsc{Case 3}: Analogous considerations can be done by combining case 1 and case 2.
\end{proof}}
{}

We next define the finite database $\RETS$ having among its subsets
a number of $\Sigma$-repairs sufficient for solving CQA.

%

\begin{definition}\label{def:rets}
Let $c$ be a value in $\Gamma - \vals{\RET}$.
Consider the largest (possibly inconsistent) database, say $C$, constructible
on the domain $\vals{\RET} \cup \{c\}$ such that
$f \in C$ iff the value $c$ does not appear in the key of $f$.
Let $\sfsknulls$ be a fixed set of values arbitrarily chosen from $\Gamma - \vals{\RET}$
whose cardinality is equal to the number of occurrences of $c$ in $C$.
We denote by $\RETS$ 
one possible database for $\G$
obtained from $C$ by replacing each occurrence of $c$ with a value from $\sfsknulls$
in such a way that each value in $\sfsknulls$ occurs exactly once in $\RETS$.
($|C| = |\RETS|$.) \qed
\end{definition}

\longShort{\medskip}{}

For example, if $\vals{\RET} = \{1,2\}$ and $\G = \{p\}$ with $arity(p) = 2$
and $key(p) = \{1\}$, then $C = \{p(1,1), p(1,2), p(1,c), p(2,1), p(2,2), p(2,c)\}$.
Let us fix $\sfsknulls = \{\varepsilon_1, \varepsilon_2\}$.
Thus, $\RETS$ has the following form: $\{p(1,1), p(1,2), p(1,\varepsilon_1), p(2,1), p(2,2), p(2,\varepsilon_2)\}$.

\begin{proposition}\label{prop:bigFormulas}
The following hold:
\begin{itemize}
  \item $|\mathcal{N}| = \sum_{g \in \G}(arity(g)-|key(g)|) \cdot |\vals{\RET}|^{|key(g)|} \cdot (|\vals{\RET}| + 1)^{arity(g)-|key(g)|-1}$
  \item $|\RETS| \leq \sum_{g \in \G} (|\vals{\RET}| + 1)^{arity(g)} \leq \sum_{g \in \G} (arity(g) \cdot |\RET| + 1)^{arity(g)}$
\end{itemize}
\end{proposition}

\begin{lemma}\label{lemma:DstarHasAnHomo}
If $\B$ is a $\Sigma$-repair for $\RET$,
then there exists $\B' \in \homo(\B)$ such that $\B' \subseteq \RETS$.
\end{lemma}

\longShort{\begin{proof}
$\B'$ can be obtained from $\B$ by replacing each fact $r(t_1) \in \B$
with the unique fact $r(t_2) \in \RETS$ such that for each $i \in arity(r)$ either
$t_2^i = t_1^i$, if $t_1^i \in \vals{\RET}$, or
$t_2^i \in \sfsknulls$, if $t_1^i \not\in \vals{\RET}$.
Moreover, note that, since $\B$ cannot contain two facts with the same key and since keys only have values from $\vals{\RET}$,
then each fact in $\RETS$ can replace at most one fact in $\B$.
Finally, $\B' \in \homo(\B)$ by Definition \ref{def:homoB}.
\end{proof}}
{}

\begin{lemma}\label{lemma:ans}
Let $\B$ be a $\Sigma$-repair for $\RET$,
$\B' \in \homo(\B)$,
$q$ be a query, and
$t$ be a tuple of values from $\vals{\RET}$.
If $t \in ans(q,\B')$, then $t \in ans(q,\B)$.
\end{lemma}

{\begin{proof}
Let $q_i$ be one of the conjunctions in $q$, if $t \in ans(q_i,\B')$, then there is a substitution
$\mu'$ from the variables of $q_i$ to values in $\Gamma$
such that $\B' \models q_i(t)$.
But since, by Definition \ref{def:homoB}, each fact in $\B'$
is univocally associated with a unique fact in $\B$ by preserving the values in $\vals{\RET}$,
and since all the extra values in $\B'$ are distinct, then
there must also be a substitution $\mu$ such that $\B \models q_i(t)$.
In particular, let $x$ be a variable in $q_i$, we can define $\mu$ in such a way that
$\mu(x) = h_{\B,\B'}(\mu'(x))$, where $h$ is the homomorphism from $\B'$ to $\B$ (see Definition \ref{def:homoB}).
Clearly, if $t \in ans(q_i,\B')$ for at least one $q_i$ in $q$ then $t \in ans(q,\B')$ too and, consequently,
$t \in ans(q,\B)$
\end{proof}}
{}

The next theorem states the decidability of CQA under both the loosely-exact and the loosely-sound semantics with KDs and SFSKs only.

\begin{theorem}\label{thm:AnsExactFinite}
Let $\B$ be a $\Sigma$-repair for $\RET$,
$q$ 
a query,
and $t$ a tuple from $\vals{\RET}$.
Let $\mathbb{B}  \subseteq 2^{\RETS}$ denote the set of all $\Sigma$-repairs
contained in $\RETS$. Then,
$
    t \in ans_{\Sigma}(q,\G,\RET) \ \ \textbf{iff} \ \ t \in ans(q,\B) \ \ \forall \B \in \mathbb{B}.
$
\end{theorem}

{\begin{proof}
($\Rightarrow$) We have to prove that, \emph{if $t \in ans_{\Sigma}(q,\G,\RET)$, then $t \in ans(q,\B)$ for each $\B \in \mathbb{B}$},
or equivalently \emph{if $t \not\in ans(q,\B)$ for some $\B \in \mathbb{B}$, then $t \not\in ans_{\Sigma}(q,\G,\RET)$}.
This follows, by the definition of $ans_{\Sigma}(q,\G,\RET)$ and from the
fact that $\mathbb{B}$ only contains $\Sigma$-repairs.

\longShort{\medskip}{}

\noindent ($\Leftarrow$) We have to prove that, \emph{if $t \in ans(q,\B)$ for each $\B \in \mathbb{B}$, then $t \in ans_{\Sigma}(q,\G,\RET)$.}
Assume that $t \in ans(q,\B)$ for each $\B \in \mathbb{B}$ but $t \not\in ans_{\Sigma}(q,\G,\RET)$.
This would entail that there is a repair $\B_0$ such that $t \not\in ans(q,\B_0)$.
But,
since $t \not\in ans(q,\B')$ for each $\B' \in \homo(\B_0)$ (by Lemma \ref{lemma:ans}), and
since $\mathbb{B} \cap \homo(\B_0)$ always contains a repair, say $\B''$ (by Lemma \ref{lemma:DstarHasAnHomo}),
then we have a contradiction since $t \not\in ans(q,\B'')$ has to hold
whereas we have assumed that $t \in ans(q,\B)$ for each $\B \in \mathbb{B}$.
\end{proof}}
{}

Decidability and complexity results, under KDs and SFSKs only, follow from \mbox{Theorem \ref{thm:AnsExactFinite}.}

\begin{corollary}\label{cor:decidability}
Let $\G$ be a global schema containing KDs and SFSKs only,
$\RET$ be a possibly inconsistent database for $\G$,
$q$ be a  query,
$\Sigma \in \{\emph{loosely-exact}, \ \emph{loosely-sound}\}$,
and $t$ be a tuple of values from $\vals{\RET}$.
The problem of establishing whether $t \in ans_{\Sigma}(q,\G,\RET)$ is in $\Pi_2^p$ in data complexity.
\end{corollary}

{\begin{proof}
It suffices to prove that the problem of establishing whether $t \not\in ans_{\Sigma}(q,\G,\RET)$ is in $\Sigma_2^p$.
This can be done by (i) building $\RETS$, and (ii) guessing $\B \in 2^{\RETS}$ such that
$\B$ is a $\Sigma$-repair and $t \not\in ans(q,\B)$.
Since, by Proposition \ref{prop:bigFormulas},
$|\RETS| \in \mathcal{O}(|\RET|^\alpha)$ where $\alpha = \textsf{max}\{arity(g) : g \in \relNam{\G}\}$, then
step (i) (enumerate the facts of $\RETS$) can be done in polynomial time.
Since checking that $t \not\in ans(q,\B)$ can be done in $\PTIME$.
It remains to show that checking whether $\B$ is a $\Sigma$-repair can be done in $\CONP$.

\noindent \textbf{[loosely-exact]} If $\Sigma = \emph{loosely-exact}$, this task corresponds to checking that there is no consistent $\B' \subseteq \RET \cup \B$ such that $\B' \ominus \RET \subset \B \ominus \RET$,
where this last task is doable in $\PTIME$.

\noindent \textbf{[loosely-sound]} If $\Sigma = \emph{loosely-sound}$, this task corresponds to checking that there is no consistent $\B' \subseteq \RETS$ such that $\B' \cap \RET \supset \B \cap \RET$,
where this last task is doable in $\PTIME$.

Then the thesis follows.
\end{proof}}
{}

\subsection{Equivalence of CQA under loosely-exact and CM-complete semantics}

In this section we define some relevant cases in which CQA under loosely-exact and CM-complete semantics
coincide. 

\disclaimer

\begin{lemma}\label{lemma:eachCMrepIsLErep}
Given a database $\RET$ for a schema $\G$,
if $\B$ is a CM-complete repair for $\RET$,
then it is a loosely-exact repair for $\RET$.
\end{lemma}

\longShort{\begin{proof}
Suppose that $\B$ is a \emph{CM-complete} repair for $\RET$
(so, it is consistent w.r.t. $\G$), but it is not a \emph{loosely-exact} one.
This means that its symmetric difference with $\RET$ can be still reduced.
But, by definition of \emph{CM-complete} semantics,
$\B$ does not contain anything else but tuples in $\RET$,
namely $\B - \RET = \emptyset$.
So, the only way for ``improving'' it is to extend it with tuples from $\RET$.
But, this is not possible because $\B$ is already maximal due to the \emph{CM-complete} semantics,
namely the addition of any other tuple would violate at least one IC.
\end{proof}}
{}

\begin{corollary}\label{coroll:LEsubeqCM}
$ans_{loosely-exact}(q,\G,\D) \subseteq ans_{CM-complete}(q,\G,\D)$
\end{corollary}

\longShort{\begin{proof}
This directly follows by Lemma \ref{lemma:eachCMrepIsLErep}
in light of Definition \ref{cqa-def}.
\end{proof}}
{}

\begin{theorem}
There are cases where $ans_{loosely\textrm{-}exact}(q,\G,\D) \subset ans_{CM\textrm{-}complete}(q,\G,\D)$
\end{theorem}

\begin{proof}
By \citeN{ChomickiMarcinkowski05}, stating that the two semantics are different,
and by Corollary \ref{coroll:LEsubeqCM}.
\end{proof}

\begin{proposition}\label{stillConsIfWeAddOrDeleteFacts}
Let $\B$ be a database consistent w.r.t. a set of ICs $C$.
\begin{enumerate}
  \item  If $C$ are DCs only, then each $\B' \subset \B$
    is consistent w.r.t. $C$, as well.

  \item If $C$ are INDs only, then $\B \cup \B'$
    is consistent w.r.t. $C$ for each $\B'$
    consistent w.r.t. $C$.
\end{enumerate}
\end{proposition}
\longShort{\begin{proof}
$(1)$ Deletion of tuples can not introduce new DCs violations.

\longShort{\medskip}{}

\noindent $(2)$ Let $r(t)$ be a fact in $\B'$.
Let $d_1$ be an IND of the form $r_1 \rightarrow r$ ($r \neq r_1$).
Clearly, $r(t)$ cannot violate $d_1$ in any database
because $r$ is in the righthand side of $d_1$.
In particular, $r(t)$ cannot violate $d_1$ in $\B \cup \B'$.
Let $d_2$ be an IND of the form $r \rightarrow r_2$ (possibly, $r = r_2$).
Since $r(t)$ does not violate $d_2$ in $\B'$,
then it cannot violate $d_2$ in $\B \cup \B'$.
\end{proof}}
{}




\begin{theorem}\label{thm:someLErepContCMrep}
Given a database $\RET$ for a schema $\G$,
let $\B$ be a loosely-exact repair for $\RET$,
and $\overline{\B} = \B \cap \RET$.
\longShort{There is a CM-complete repair $\B' \subseteq \overline{\B}$ for $\RET$
if at least one of the following restrictions holds:
\begin{enumerate}
  \item[I] $\G$ contains DCs only (no INDs);
  \item[II] $\G$ contains INDs only (no DCs);
  \item[III] $\G$ contains KDs and FKs only, and $\RET$ is consistent w.r.t. KDs;
  \item[IV] $\G$ contains KDs and SFKs only;
\end{enumerate}}
{There is a CM-complete repair $\B' \subseteq \overline{\B}$ for $\RET$
if at least one of the following holds:
\emph{(I)} $\G$ contains DCs only (no INDs);
\emph{(II)} $\G$ contains INDs only (no DCs);
\emph{(III)} $\G$ contains KDs and FKs only, and $\RET$ is consistent w.r.t. KDs;
\emph{(IV)} $\G$ contains KDs and SFKs only.}
\end{theorem}

\longShort{\begin{proof}
\textbf{Case I}: By Proposition \ref{stillConsIfWeAddOrDeleteFacts},
since $\B$ is consistent w.r.t. DCs, then $\overline{\B} \subseteq \B$
is consistent as well.
Now, if $\B - \RET \neq \emptyset$, then we would have a contradiction because
$\overline{\B} \ominus \RET \subset \B \ominus \RET$ would hold.
Thus, $\B - \RET = \emptyset$ and so,
$\B = \overline{\B}$ is already a \emph{CM-complete} repair itself.

\longShort{\medskip}{}

\noindent \textbf{Case II}:
Since there is no DC, there exists only one \emph{CM-complete} repair,
say $\B'$, obtained from $\RET$ after removing all the facts violating INDs.
Now, if $\B'$ was not contained in $\B$,
then, by Proposition \ref{stillConsIfWeAddOrDeleteFacts},
$\B' \cup \B$ would still be consistent,
that is a larger \emph{CM-complete} repair. Contradiction.
Finally $\overline{\B} = \B'$.

\longShort{\medskip}{}

\noindent  \textbf{Case III}: Since $\RET$ is consistent w.r.t. DCs,
we have only one \emph{CM-complete} repair, say $\B'$,
obtained from $\RET$ after removing all the facts violating INDs.
But, as in case II, if the set $\B' - \overline{\B}$ was nonempty,
then we could add all these facts into $\B$ without violating any IND.
Anyway, one of these facts, say $f$,
could violate a DC due to a fact $f'$ in $\B - \RET$.
Now, note that $f'$ is in $\B$ only for fixing an IND violation.
But in this case, as we are only considering FKs,
there would be no reason to have $f'$ in $\B$ instead of $f$.
So, we could (safely) replace $f$ with $f'$ in $\B$
and no KD would be violated as well as no FK.
But this leads to a contradiction.
So, there is no fact in $\B'$ which is not in $\overline{\B}$.

\longShort{\medskip}{}

\noindent  \textbf{Case IV}: First of all, we observe that if $\B - \RET = \emptyset$,
then either $\overline{\B}$ is a \emph{CM-complete} repair or $\B$ is not a \emph{loosely-exact} repair.
So the statement holds.
Now assume that $\B - \RET \neq \emptyset$.
We distinguish three different cases:
\begin{enumerate}
\item[] \begin{enumerate}
  \item[(1)] $\overline{\B}$ is both consistent and maximal (it is a \emph{CM-complete} repair);
  \item[(2)] $\overline{\B}$ is consistent but not maximal (it is not a \emph{CM-complete} repair);
  \item[(3)] $\overline{\B}$ is inconsistent (it is not a \emph{CM-complete} repair).
\end{enumerate}
\end{enumerate}

In case $(1)$, we have a contradiction because $\B$ is assumed to be a \emph{loosely-exact} repair,
but it does not minimize the symmetric difference with $\RET$
since $\overline{\B} \ominus \RET \subset \B \ominus \RET$.

\longShort{\medskip}{}

In case $(2)$, we have again a contradiction because $\B$ is assumed to be a \emph{loosely-exact} repair
but it does not minimize the symmetric difference with $\RET$
since there is a \emph{CM-complete} repair $\widetilde{\B} \supset \overline{\B}$
such that $\widetilde{\B} \ominus \RET \subset \B \ominus \RET$.

\longShort{\medskip}{}

In case $(3)$, we observe that since, by hypothesis, $\B$ is consistent,
then the inconsistency of $\overline{\B}$ arises, by Proposition \ref{stillConsIfWeAddOrDeleteFacts},
only due to INDs.
Now, assume that (i) $\overline{\B}$ contains a fact $r_1(t_1)$; 
(ii) there is an IND $d$ of the form \indNorm{r}; 
(iii) there is no fact for $r_2$ in $\overline{\B}$ satisfying $d$.
This means that a fact of the form $r_2(t_2)$ must be in $\B - \RET$,
where $t_1^{\pi_L^d} = t_2^{\pi_R^d}$. 

Now, we claim that there is no fact of the form $r_2(t_3)$ in $\RET - \overline{\B}$,
where $t_1^{\pi_L^d} = t_3^{\pi_R^d}$.
Suppose that $\RET - \overline{\B}$ contained such a fact $r_2(t_3)$.
Consider the new database $(\B \cup \{r_2(t_3)\}) - \{r_2(t_2)\}$.
This would necessarily be consistent because the addition of $r_2(t_3)$
(after removing $r_2(t_2)$ as well) cannot violate any KD
since $d$ is an FK (remember that $key(r_2) = \pi_R^d$),
and cannot violate any IND since each IND $d'$ of the form
$r_2 \rightarrow r_3$ is an SFK (remember that $key(r_2) \supseteq \pi_L^{d'}$).
But this is not possible because $\B$ is assumed to be a \emph{loosely-exact} repair,
and $(\B \cup \{r_2(t_3)\}) - \{r_2(t_2)\}$
would improve the symmetric difference.
This means, that each \emph{CM-complete} repair cannot contain the tuple $r_1(t_1)$
(this goes in the direction of the statement).

Let us call $\overline{\B}'$ the consistent (w.r.t. both KDs and SFKs) database
obtained from $\overline{\B}$
after removing all the facts violating some IND.
It remains to show that there is no other fact in $\RET - \overline{\B}$ such that
$\overline{\B}' \cup \{r_1(t_1)\}$ does not violate any constraint.
Assume that such a fact $r_1(t_1)$ exists, then:
\begin{enumerate}
\item[] \begin{enumerate}
  \item[-] $\overline{\B}' \cup \{r_1(t_1)\}$ would not violate any IND;

  \item[-] $\B  \cup (\overline{\B}' \cup \{r_1(t_1)\}) = \B \cup \{r_1(t_1)\}$
    would not violate any IND, by Proposition \ref{stillConsIfWeAddOrDeleteFacts};

  \item[-] $\B \cup \{r_1(t_1)\}$ would violate some KD,
    since $\B$ is a \emph{loosely-exact} repair.
\end{enumerate}
\end{enumerate}
Thus, there would necessarily be a fact in $\B$, say $r_1(t_2)$, being not in $\overline{\B}'$,
with the same key of $r_1(t_1)$. Since such a fact cannot stay in $\overline{\B} - \overline{\B}'$
because it does not violate any IND, then it must be in $\B - \RET$. But this is not possible
because we could replace $r_1(t_2)$ by $r_1(t_1)$ in $\B$ without violating any KD
and also without violating any IND, since we are only considering SFKs.
But since $\B$ is already a repair, this is clearly a contradiction.
Finally, $\overline{\B}'$ is a \emph{CM-complete} repair.
\end{proof}}
{}


\longShort{\begin{corollary}
\label{cor:equivalence}
$ans_{loosely\emph{-}exact}(q,\G,\D) = ans_{CM\emph{-}complete}(q,\G,\D)$ in the following cases:
\begin{itemize}
  \item[-] $\G$ contains DCs only (no INDs);
  \item[-] $\G$ contains INDs only (no DCs);
  \item[-] $\G$ contains KDs and FKs only, and $\RET$ is consistent w.r.t. KDs;
  \item[-] $\G$ contains KDs and SFKs only;
\end{itemize}
\end{corollary}}
{\begin{corollary}
\label{cor:equivalence}
$ans_{loosely\emph{-}exact}(q,\G,\D) = ans_{CM\emph{-}complete}(q,\G,\D)$ in the cases where
Theorem \ref{thm:someLErepContCMrep} holds.
\end{corollary}}

\longShort{\begin{proof}
This directly follows by both Theorem \ref{thm:someLErepContCMrep} and
Lemma \ref{lemma:eachCMrepIsLErep},
in light of Definition \ref{cqa-def}.
\end{proof}}
{}

\begin{proposition}
In general, Theorem~\ref{thm:someLErepContCMrep} does not hold in case $\G$ contains SFSKs and KDs only.
\end{proposition}

\longShort{\begin{proof}
Consider a database containing two relations of arity 2, namely:
$r$ and $s$. Moreover, the schema contains the following ICs:
$key(r) = \{1,2\}$, and $key(s) = \{1\}$ and $r(X,Y) \rightarrow s(X,Y)$.
Note that, the last is a safe FSK.
Suppose also that a DB $\RET$ for  this schema contains the following facts: $r(a,b)$, $s(a,c)$.
The {\em loosely-exact} repairs are $\B_1 = \{s(a,c)\}$ and $\B_2 = \{r(a,b), s(a,b)\}$, but only the first one is
also a {\em CM-Complete} repair.
However,  $\overline{\B} = \B_2 \cap \D = \{r(a,b)\}$ is not a CM-complete repair (it is inconsistent).
The only consistent database contained in $\overline{\B}$ is the empty set
that is not a {\em CM-Complete} repair (deletions are not minimized).
\end{proof}}
{}

\section{Computation of CQA via ASP}\label{sec:cqa}

In this section, we show how to exploit \emph{Answer Set Programming}
(ASP)~\cite{GelfondLifschitz88,GelfondLifschitz91} for efficiently computing consistent answers
to user queries under different semantic assumptions. ASP is a powerful logic programming
paradigm allowing (in its general form) for disjunction in rule heads~\cite{Minker82} and
nonmonotonic negation in rule bodies.
In the following, we assume that the reader is familiar with ASP with aggregates, and in particular we adopt the
DLV syntax \cite{fabe-etal-2008-aij,leon-etal-2002-dlv}.

The suitability of ASP for implementing CQA has been already recognized in the
literature~\citecqa. The general approaches are based on the following idea: produce an ASP
program $P$ whose answer sets represent possible repairs, so that the problem of computing CQA
corresponds to cautious reasoning on $P$. One of the hardest challenges in this context is the
automatic identification of a program $P$ considering a minimal number of repairs actually relevant to
answering user queries.



In order to face these challenges, we first introduce a general encoding which unifies in a
common core the solutions for CQA under the semantics considered in this paper. Then, based on this unified
framework, we define optimization strategies precisely aiming at reducing the computational cost of CQA.
This is done in several ways: {\em (i)} by casting down the original program to
complexity-wise easier programs; {\em (ii)} by identifying
portions of the database not requiring repairs at all, according to the query requirements; {\em
(iii)} exploiting equivalence classes between some semantics in such a way to adopt optimized solutions.

We next present the general encoding first and, then, the optimizations.

\subsection{General Encoding}
\label{sub:generalencoding}

The general approach generates a program $\pcqa$ and a new query $\qcqa$ obtained by rewriting
both the constraints and the query $q$ in such a way that CQA reduces to cautious reasoning
on $\pcqa$ and $\qcqa$. Recall that a union of conjunctive queries in ASP is expressed as a set of
rules having the same head predicate with the same arity.

In what follows, we first present how to generate $\pcqa$ and $\qcqa$ and then formally prove
under which hypothesis cautious reasoning on such $\pcqa$ and $\qcqa$ corresponds to CQA.

Given a database $\RET$ for a schema $\G$ and a query $q$ on $\G$, the ASP program
$\pcqa$
is created by rewriting each IC belonging to $\intCon{\G}$ and $q$ as follows:

\paragraph{Denial Constraints.}

Let $\Sigma \in \{$\emph{CM-complete}, \emph{loosely-sound}, \emph{loosely-exact}$\}$. For each DC of the form $\forall
\vars{x}_1, \ldots, \vars{x}_m \ \ \neg [g_1(\vars{x}_1) \wedge \ldots \wedge g_m(\vars{x}_m)
\wedge \sigma(\vars{x}_1, \ldots, \vars{x}_m)]$ in $\intCon{\G}$, insert the following rule into
$\pcqa$:
\begin{itemize}
  \item $g_1^{c}(\vars{x}_1) \Or \cdots \Or g_m^{c}(\vars{x}_m) \derives g_1(\vars{x}_1), \cdots, g_m(\vars{x}_m), \sigma(\vars{x}_1, \ldots, \vars{x}_m).$
\end{itemize}

This rule states that in presence of a violated denial constraint it must be guessed the tuple(s) to
be removed in order to repair the database.


\paragraph{Inclusion dependencies.}

Let $\Sigma = \{$\emph{CM-complete}, \emph{loosely-exact}$\}$. For each IND $d$ in $\intCon{\G}$
of the form \indNorm{g}, add the following rules into $\pcqa$:
\begin{itemize}
  \item $g_1^c(\vars{x}_1) \derives g_1(\vars{x}_1), \
    \textsf{\#count} \{\vars{x}_{2\exists} : g_2^c(\vars{x}_2)\} =
    \textsf{\#count} \{\vars{x}_{2\exists} : g_2(\vars{x}_2)\}.$ \hfill if $|\vars{x}_{2\exists}| > 0$

  \item $g_1^c(\vars{x}_1) \derives g_1(\vars{x}_1), \ g_2^c(\vars{x}_2).$ 

  \item[] $g_1^c(\vars{x}_1) \derives g_1(\vars{x}_1), \ not \ g_2(\vars{x}_2).$ \hfill if $|\vars{x}_{2\exists}| = 0$
\end{itemize}


The first rule states that a tuple of $g_1$ must be deleted iff either all the tuples in $g_2$
previously referred to by $g_1$ via $d$ have been deleted due to the repairing process, or
there is no tuple in $g_2$ referred to by $g_1$ via $d$. (This is done by comparing the total count of tuples in
$g_2$ and $g_2^c$).
%
%
Observe that if there is a cyclic set of INDs, the set of rules generated by this rewriting would
contain recursive aggregates. Their semantics is described in \cite{fabe-etal-2008-aij}.
The latter two rules replace the first one in the special case of $|\vars{x}_{2\exists}| = 0$.

\paragraph{Repaired Relations.}

Let $\Sigma \in \{$\emph{CM-complete}, \emph{loosely-sound}, \emph{loosely-exact}$\}$.
For each relation name $g \in \relNam{\G}$, insert the following rule into $\pcqa$:
\begin{itemize}
\item $\rep{g}(\vars{x}) \derives g(\vars{x}), \ \naf\ g^c(\vars{x}).$
\end{itemize}

\paragraph{Query rewriting.}

Build $\qcqa(\vars{x})$ from $q(\vars{x})$ as follows:
\begin{enumerate}
  \item If $\Sigma = \emph{loosely-sound}$, then apply onto $q$
    the perfect rewriting algorithm that deals with INDs described in
    \cite{CaliLemboRosati03}\footnote{Observe that,
    when $\Sigma = \emph{loosely-sound}$, INDs are not encoded into logic rules.}.

  \item For each atom $g(\vars{y})$ in $q$, replace $g(\vars{y})$ by $\rep{g}(\vars{y})$ 

\end{enumerate}

The perfect rewriting introduced in \cite{CaliLemboRosati03} is intuitively described next.
Given a query $q(\vars{x})$ and a set of INDs, the algorithm iteratively computes a new query
$Q$ as follows. $Q$ is first initialized with $q$; then, at each iteration
it carries out the following two steps: {\em (1)} For
each conjunction $cq'$ in $Q$, and for each pair of atoms $g_1$, $g_2$ in $cq'$ that unify (i.e., for which there
exists a substitution transforming $g_1$ into $g_2$), $g_1$ and $g_2$ are substituted by one
single unifying atom.
{\em (2)} For each conjunction $cq'$ in $Q$, and for each {\em applicable} IND $d$ of the form
${g}_1 \rightarrow {g}$ such that $g$ is in $cq'$, it adds to $Q$ a new conjunction $cq''$ obtained
from $cq'$ by interpreting $d$ as a rewriting rule on $g$, applied from right to left.
The algorithm stops when no further modifications are possible on $Q$ with the two steps above.

The following theorems show how and when cautious reasoning on $\pcqa$ and $\qcqa$ correspond
to CQA. First we consider the CM-complete semantics.

\disclaimer

\begin{theorem}
\label{th:cm}
%
Let $\Sigma = \emph{CM-complete}$, let $\RET$ be a database
for a schema $\G$ with arbitrary DCs and (possibly cyclic) INDs, and
let $q$ be a union of conjunctive queries. $t \in ans_{\Sigma}(q,\G,\D)$ iff
$\qcqa(t)$ is a cautious consequence of the ASP program $\RET \cup \pcqa$.
\end{theorem}

\longShort{\begin{proof}
We claim that $\pcqa$ allows to consider only and all the repairs, exactly one per model.
Let $\B^r$ be a repair.
In the following, we describe how to obtain a model containing for each relation, say $g$,
exactly only and all the tuples of $g$ that do not appear in $\B^r$.
We collect such tuples in the new relation $g^c$,
while we collect in $g^r$ only and all the tuples of $g$ appearing in $\B^r$.
%
For each relation, say $g$:
\begin{enumerate}
\item[] \begin{enumerate}
  \item By the disjunctive rules (if any) involving $g$, of the form
    \[
        \cdots \Or g^c(\vars{x}) \Or \cdots \ \derives \
        \cdots, \ g(\vars{x}), \ \cdots, \ \sigma(\cdots, \vars{x}, \cdots).
    \]
    we guess a set of tuples of $g$,
    collected in $g^c$, that must not appear in $\B^r$.

  \item Next, for each IND of the form $g(\vars{x}_1) \rightarrow g_1(\vars{x}_2)$
    (involving $g$ in the left-hand side), we use the rule
    \[
        g^c(\vars{x}_1) \derives g(\vars{x}_1), \
        \textsf{\#count} \{\vars{x}_{2\exists} : g_1^c(\vars{x}_2)\} =
        \textsf{\#count} \{\vars{x}_{2\exists} : g_1(\vars{x}_2)\}.
    \]
    for deciding which tuples of $g$ cannot appear in $\B^r$ due to an IND violation.
    Note that in case $|\vars{x}_{2\exists}|=0$, the rule is rewritten without the \textsf{\#count} aggregate.

  \item Finally, by the rule
$      \rep{g}(\vars{x}) \derives g(\vars{x}), \ \naf\ g^c(\vars{x})$
  we obtain the repaired relations.
\end{enumerate}
\end{enumerate}
Importantly, for computing the extension of each $g^c$ we only
exploit the minimality of answer sets semantics;
later, the extension of each $g^r$ is computed. Observe that, by the splitting theorem \cite{LifschitzTurner94}
$\pcqa$ can be divided (split) into two parts .
It is clear that, by construction, $\pcqa$ has exactly one answer set per repair.
Finally, the query is reorganized to exploit the repaired relations,
and cautious reasoning does the rest.
\end{proof}}
{}

\begin{example}
\noindent Consider again Example \ref{ex-repairs1}, the  program (and the query built from
$q(X) \derives m(X)$) under the \emph{CM-complete} semantics obtained for it, is:

\begin{enumerate}
  \item[$\centerdot$] $\can{e}(X_c,X_n) \Or \can{e}(X_c,X_n') \derives e(X_c,X_n), \ e(X_c,X_n'), \ X_n \neq X_n'.$

  \item[$\centerdot$] $m^c(X_c) \derives m(X_c), \ \textsf{\#count} \{X_n' : e^c(X_c,X_n')\} = \textsf{\#count} \{X_n : e(X_c,X_n)\}.$

  \item[$\centerdot$] $e^r(X_c,X_n) \derives e(X_c,X_n),\ \naf\ e^c(X_c,X_n).$

  \item[$\centerdot$] $m^r(X_c) \derives m(X_c),\ \naf\ m^c(X_c).$

  \item[$\centerdot$] $\qcqa(X_c) \derives \rep{m}(X_c).$
\end{enumerate}

\noindent When this program is evaluated on the database we obtain four answer sets. It
can be verified that, all the answer sets contain $\rep{m}(\emph{`e1'})$ and
$\rep{m}(\emph{`e2'})$, (i.e., they are cautious consequences of \pcqa) and, thus, $\emph{`e1'}$
and $\emph{`e2'}$ are the consistent answers to the query. \qed
\end{example}

\begin{theorem}
\label{th:ls}
Let $\Sigma = \emph{loosely-sound}$, let $\RET$ be a database
for a schema $\G$ with KDs (and exactly one key for each relation) and (possibly cyclic) NKC INDs, and
let $q$ be a union of conjunctive queries without comparison atoms\footnote{
Recall that equalities are expressed in terms of variables having the same name.}. $t \in ans_{\Sigma}(q,\G,\RET)$ iff
$\qcqa(t)$ is a cautious consequence of the ASP program $\RET \cup \pcqa$.
\end{theorem}

\longShort{\begin{proof}

Considerations analogous to the \emph{CM-complete} case can be drawn.
Disjunctive rules guess a minimal set of tuples to be removed,
whereas the perfect rewriting algorithm allows to deal with NKC INDs.
Observe that, the separation theorem introduced in \cite{CaliLemboRosati03}
shows that INDs can be taken into account as if the KDs where not expressed on $\G$;
in particular, it states that it is sufficient to compute the perfect rewriting $q'$
of $q$ and evaluate $q'$ on the maximal subsets of $\RET$ consistent with KDs.
In our case, these are computed by the part of $\pcqa$ dealing with KDs, whereas
the separation is carried out by renaming each $g$ in $q'$ by $g^r$.
\end{proof}}
{}

The general encoding for the \emph{loosely-exact} semantics is inherently
more complex than the ones for \emph{loosely-sound}
and \emph{CM-complete}, since both tuple deletions and tuple insertions are subject to minimization.
As a consequence, we tackled the \emph{loosely-exact} encoding by considering that
there are common cases in which CQA under the \emph{loosely-exact} semantics and the \emph{CM-complete} semantics
actually coincide (see Corollary \ref{cor:equivalence}). These cases can be easily checked and, thus, it is possible to handle
the \emph{loosely-exact} semantics with the  encoding defined for the \emph{CM-complete} case.

\longShort{\begin{theorem}
\label{th:le}
Let $\Sigma = \emph{loosely-exact}$, $\RET$ be a database
for a schema $\G$ such that one of the following holds:
\begin{itemize}
  \item[-] $\G$ contains DCs only (no INDs);
  \item[-] $\G$ contains INDs only (no DCs);
  \item[-] $\G$ contains KDs and FKs only, and $\RET$ is consistent w.r.t. KDs;
  \item[-] $\G$ contains KDs and SFKs only;
\end{itemize}
\noindent Let $q$ be a union of conjunctive queries. 
$t \in ans_{\Sigma}(q,\G,\RET)$ iff $\qcqa(t)$ is a cautious consequence of the ASP program $\RET \cup \pcqa$.
\end{theorem}}
{\begin{theorem}
\label{th:le}
Let $\Sigma = \emph{loosely-exact}$, $\RET$ be a database
for a schema $\G$ such that one of the following holds:
\emph{(i)} $\G$ contains DCs only (no INDs);
\emph{(ii)} $\G$ contains INDs only (no DCs);
\emph{(iii)} $\G$ contains KDs and FKs only, and $\RET$ is consistent w.r.t. KDs;
\emph{(iv)} $\G$ contains KDs and SFKs only.
Given a union of conjunctive queries $q$, then
$t \in ans_{\Sigma}(q,\G,\RET)$ iff $\qcqa(t)$ is a cautious consequence of the ASP program $\RET \cup \pcqa$.
\end{theorem}}


\begin{proof} Follows from Corollary
\ref{cor:equivalence} and Theorem \ref{th:ls}.
\end{proof}

\subsection{Optimized Solution}\label{subsec:optimizedsolution}\label{sub:l-e}

The strategy reported in the previous section is a general solution for solving the CQA problem
but, in several cases, more efficient ASP programs can be produced.
First of all, note that the general algorithm blindly considers all the ICs on the global schema,
including those that have no effect on the specific query. Consequently, useless logic rules might
be produced which may slow down program evaluation. Then, a very simple optimization may consist
of considering relevant ICs only. However,  there are several cases in which
the complexity of CQA stays in $\PTIME$;
but disjunctive programs, for which cautious reasoning becomes a hard
task~\cite{EiterGottlobMannila97}, are generated even in presence of denial constraints only.
This means that the evaluation of the produced logic programs might be much more expensive than required
in those ``easy'' cases. In the following, we provide semantic-specific optimizations aiming
to overcome such problems for the settings pointed out in Theorem~\ref{th:cm}, Theorem~\ref{th:ls}, and Theorem~\ref{th:le}.

Given a query $q$ and an atom $g$ in  $q$,
we define the set of \emph{relevant indices} of $g$ in $q$, say $relevant(q,g)$ in such a way that
an index $i$ in [1..$arity(g)$] belongs to $relevant(q,g)$ if at least one of the following holds
for an occurrence $g(X_1,\ldots,X_n)$ of $g$ in $q$:
\begin{itemize}
  \item $X_i$ is not existentially quantified (it is a free variable, it
    is an output variable of $q$);
  \item $X_i$ is involved in some comparison atom (even if it is existentially quantified);

  \item $X_i$ appears more than once in the same conjunction;

  \item $X_i$ is a constant value;
\end{itemize}
\noindent If $g$ does not appear in  $q$, we say that $relevant(q,g) = \emptyset$;

In the following, we denote by $\pi$ a set of indices. Moreover,
given a sequence of variables $\vars{x}$ and a set $\pi \subseteq \{1,\ldots,|\vars{x}|\}$,
we denote by $\vars{x}^{\pi}$ the sequence obtained from $\vars{x}$
by discarding a variable if its position is not in $\pi$.
Finally, given a relation name $g$, a set of indices $\pi$ and a label $\ell$ we denote by
$g^{\ell\textit{-}\pi}(\vars{x}^{\pi})$ an auxiliary atom derived from $g$, marked by $\ell$, and using
only variables in $\vars{x}^{\pi}$.

\paragraph{\emph{\textbf{$\Sigma = \emph{loosely-sound}$}.}}
The objective of this optimization is to single out, for each relation involved by the query,
the set of attributes actually relevant to answer it and apply the necessary repairs only on them.
As we show next, this may allow both to reduce (even to zero) the number of disjunctive rules needed to
repair key violations and to reduce the cardinality of relations involved in such disjunctions.

Given a schema $\G$ and a query $q$, perform the following steps for building the program $\pcqa$
and the query $\Qcqa$.

\begin{enumerate}
  \item[1.] Apply the the perfect rewriting algorithm that deals with INDs  described in \cite{CaliLemboRosati03}.
\end{enumerate}

\begin{enumerate}
  \item[2.] Let $Q$ be the union of conjunctive queries obtained from $q$ after Step 1.
    For each $g \in \relNam{\G}$, build the sets
    \[
        \pi^g_R =  relevant(Q,g)   \quad \quad \pi^g_S = \pi^g_R \cup key(g)
    \]
\end{enumerate}

\noindent These two sets capture the fact that a key attribute is relevant for the repairing process, but it may not be
strictly relevant for answering the query.

Observe that the perfect rewriting dealing with INDs must be applied \emph{before} singling out
relevant attributes. In fact, $q$ may depend, through INDs,
also on attributes of relations not explicitly mentioned in it. However, in the last step of this algorithm
the rewriting of the query is completed by substituting each relation in the query with its repaired
(and possibly reduced) version.

\begin{enumerate}
  \item[3.] For each $g \in \relNam{\G}$ such that $\pi^g_R \neq \emptyset$ and $key(g) \nsupseteq \pi^g_R$,
    add the following rules into $\pcqa$:
    \begin{itemize}
      \item $g^{sr\textit{-}\pi^g_S}(\vars{x}^{\pi^g_S}) \derives g(\vars{x})$.

      \item $g^{c\textit{-}\pi^g_S}(\vars{x}_1^{\pi^g_S}) \Or
             g^{c\textit{-}\pi^g_S}(\vars{x}_2^{\pi^g_S}) \derives
             g^{sr\textit{-}\pi^g_S}(\vars{x}_1^{\pi^g_S}), \
             g^{sr\textit{-}\pi^g_S}(\vars{x}_2^{\pi^g_S}), \
             \vars{x}_1^{i} \neq \vars{x}_2^{i}.$ \\
             $\textcolor{white}{.} $ \hfill $\forall i \in \pi^g_S - key(g)$

      \item $g^{r\textit{-}\pi^g_R}(\vars{x}^{\pi^g_R}) \derives
        g^{sr\textit{-}\pi^g_S}(\vars{x}^{\pi^g_S}), \ \naf\
        g^{c\textit{-}\pi^g_S}(\vars{x}^{\pi^g_S})$.
    \end{itemize}
\end{enumerate}

\noindent Observe that if there exists at least one relevant
non-key attribute for $g$, the repairing process can not be avoided; however,
violations caused by irrelevant attributes only (i.e, not in $\pi^g_S$)
can be ignored, since the projection of $g$ on $\pi^g_S$ is still safe and sufficient
for query answering purposes.

\begin{enumerate}

  \item[4.] For each  $g \in \relNam{\G}$ such that $\pi^g_R \neq \emptyset$ and $key(g) \supseteq \pi^g_R$,
    add the following rule into $\pcqa$:
    \begin{itemize}
      \item $g^{r\textit{-}\pi^g_R}(\vars{x}^{\pi^g_R}) \derives g(\vars{x})$.
    \end{itemize}
\end{enumerate}

\noindent Observe that, if the relevant attributes of $g$ are a subset of its key,
the repair process of $g$ for key violations through disjunction can be avoided at all.
In fact, the projection of $g$ on $\pi^g_R$ is still safe and sufficient for query answering purposes.
Moreover, for the same reason, it is not needed to take all the key of $g$ into account.

\begin{enumerate}

  \item[5.]
  For each atom of the form $g(\vars{x})$ in $Q$,  replace $g(\vars{x})$ by $g^{r\textit{-}\pi^g_R}(\vars{x}^{\pi^g_R})$.
\end{enumerate}

\paragraph{\emph{\textbf{$\Sigma = \emph{CM-complete}$}.}}
For the optimization of the CM-complete semantics, we exploit a graph which is used to
navigate the query and the database in order to single out those relations and
projections actually relevant for answering the query.
Moreover, it allows to identify possible cycles generated by ICs which must be
suitably handled; in fact, acyclic ICs induce a partial order among them and this information
can be effectively exploited for the optimization. On the contrary cyclic ICs must be handled in a more standard way.

\longShort{Given a schema $\G$ and a query $q$, build the directed labelled graph $G_q= \tup{N,A}$ as follows:
\begin{itemize}
  \item $N = \{q\} \cup \relNam{\G}$;

  \item $(g_1,g_2,c) \in A$ iff $c$ is a DC in $\intCon{\G}$ involving both $g_1$ and $g_2$;

  \item $(g_1,g_2,d) \in A$ iff $d$ is an IND in $\intCon{\G}$ of the form  \indMicro{g};

  \item $(q,g,\varepsilon) \in A$ iff $g$ appears in a conjunction of $q$.
\end{itemize}}
{Given a schema $\G$ and a query $q$, build the directed labelled graph $G_q= \tup{N,A}$ as follows:
$N = \{q\} \cup \relNam{\G}$;
$(g_1,g_2,c) \in A$ iff $c$ is a DC in $\intCon{\G}$ involving both $g_1$ and $g_2$;
$(g_1,g_2,d) \in A$ iff $d$ is an IND in $\intCon{\G}$ of the form  \indMicro{g};
$(q,g,\varepsilon) \in A$ iff $g$ appears in a conjunction of $q$.}
\noindent Perform the following steps for building program $\pcqa$:
\begin{enumerate}
  \item[1.] Visit $G_q$ starting from node $q$;

  \item[2.] Discard  unreachable nodes and update the sets $N$ and $A$;

  \item[3.] Partition the set $N$ in $(N_{cf}, N_{ncf})$ in such a way that a node $n$
    belongs to $N_{cf}$ if it is not involved in any cycle ($q$ always belongs to $N_{cf}$).
    Contrariwise,  a node $n$ belongs to $N_{ncf}$ if it is involved in some cycle.

  \item[4.] For each node $g \in N - \{q\}$ compute the sets
    \begin{enumerate}

      \item[] $\pi_R^g = (\bigcup_{(g_L,g,d) \in A} \pi_R^d) \cup relevant(q,g)$;

      \item[]  $\pi_S^g = \pi_R^g \cup key(g)$, only if $g$ has
        exactly one primary key as DCs; $\pi_S^g= \emptyset$ otherwise.
    \end{enumerate}

    \noindent here $\pi_R^g$ is the set of relevant variable indices of $g$, and $\pi_S^g$
    adds to $\pi_R^g$ the key of $g$.

\end{enumerate}

\noindent Observe that Steps 1--4 implement a pre-processing phase in which relevant relations and
their relevant indices are singled out, and each relevant relation is classified as
cycle free or non cycle free.

\begin{enumerate}

  \item[5.] For each node $g \in N_{cf}$, if $g$ has only one key as DCs, then add the following
    rules into $\pcqa$:
    \begin{itemize}
      \item $g^{\xi\textit{-}\pi_{\chi}^{g}}(\vars{x}^{\pi_{\chi}^g}) \derives
           g(\vars{x}), \ g_1^{r\textit{-}\pi_R^{d_1}}(\vars{x}_1^{\pi_R^{d_1}}), \ \ldots, \
           g_k^{r\textit{-}\pi_R^{d_k}}(\vars{x}_k^{\pi_R^{d_k}})$.

      \item $g_i^{r\textit{-}\pi_R^{d_i}}(\vars{x}_i^{\pi_R^{d_i}}) \derives
           g_i^{r\textit{-}\pi_R^{g_i}}(\vars{x}_i^{\pi_R^{g_i}})$.
           \hfill $\forall i \in$ [1..$k$] s.t. $\pi_R^{g_i} \supset \pi_R^{d_i}$
    \end{itemize}
    where:
      \begin{itemize}
        \item[-] $k \geq 0$ is the number of arcs in $G_q$ labelled by INDs, and outgoing from $g$;

        \item[-] the pair $(\xi, \chi)$ is either $(r,R)$ or $(sr, S)$, according to whether
        $key(g) \supseteq \pi_R^{g}$ or not, respectively. Intuitively, if $key(g) \supseteq \pi_R^{g}$
        holds, then the repair $g^{r\textit{-}\pi_{R}^{g}}$ of $g$ can be directly computed;
        otherwise the computation must first go through a semi-reparation step for
        computing $g^{sr\textit{-}\pi_{S}^{g}}$. Intuitively, this semi-reparation step
        collects those tuples that violate no IND of the form $g \rightarrow g_i$,
        but that must be anyway processed in order to fix some key violation (see Steps 6 - 10).

        \item[-] atom $g_i^{r\textit{-}\pi_R^{d_i}}$  is in the body of the first rule ($1 \leq i \leq k$)
          only if both $(g,g_i,d_i) \in A$, and $d_i$ is an IND of the form
          $g(\vars{x}) \rightarrow g_i(\vars{x}_i)$. This atom is just a projection of
          $g_i^{r\textit{-}\pi_R^{g_i}}(\vars{x}_i^{\pi_R^{g_i}})$.
      \end{itemize}

  \item[6.] For each node $g \in N_{cf}$ if
    $g$ has only one primary key as DCs, and
    $key(g) \subset \pi_R^g$, and
    $g$ has incoming arcs only from $q$, and
    all the relevant variables of $g$ w.r.t. $q$ are in the head of $q$, and
    each occurrence of $g$ in $q$ contains all of its relevant variables, then
    add the following rules into $\pcqa$ by considering
    that the key of $g$ is defined by rules of the form
    $\forall \vars{x}_1, \vars{x}_2 \ \ \neg [g(\vars{x}_1) \wedge g(\vars{x}_2) \wedge \vars{x}_1^i \neq
    \vars{x}_2^i]$:
    \begin{itemize}
      \item $g^{c\textit{-}\pi^g_S}(\vars{x}_1^{\pi^g_S}) \derives
           g^{sr\textit{-}\pi^g_S}(\vars{x}_1^{\pi^g_S}), \
           g^{sr\textit{-}\pi^g_S}(\vars{x}_2^{\pi^g_S}), \
           \vars{x}_1^{i} \neq \vars{x}_2^{i}.$
           \hfill $\forall i \in \pi^g_S - key(g)$

    \item $g^{r\textit{-}\pi_R^{g}}(\vars{x}_1^{\pi_R^{g}}) \derives
        g^{sr\textit{-}\pi^g_S}(\vars{x}_1^{\pi^g_S}), \ \naf\
        g^{c\textit{-}\pi^g_S}(\vars{x}_1^{\pi^g_S})$.
      \end{itemize}

  \item[7.] For each node $g \in N_{cf}$ if
    $g$ has only one primary key as DCs, and
    $key(g) \nsupseteq \pi_R^{g}$, and
    case 6 does not apply, then add the following
    rules into $\pcqa$
    by considering that the key is defined by rules of the form,
    $\forall \vars{x}_1, \vars{x}_2 \ \ \neg [g(\vars{x}_1) \wedge g(\vars{x}_2) \wedge \vars{x}_1^i \neq \vars{x}_2^i]$:

    \begin{itemize}
      \item $g^{c\textit{-}\pi^g_S}(\vars{x}_1^{\pi^g_S}) \Or
             g^{c\textit{-}\pi^g_S}(\vars{x}_2^{\pi^g_S}) \derives
             g^{sr\textit{-}\pi^g_S}(\vars{x}_1^{\pi^g_S}), \
             g^{sr\textit{-}\pi^g_S}(\vars{x}_2^{\pi^g_S}), \
             \vars{x}_1^{i} \neq \vars{x}_2^{i}.$ \\
             $\textcolor{white}{.} $ \hfill $\forall i \in \pi^g_S - key(g)$

      \item $g^{r\textit{-}\pi_R^{g}}(\vars{x}_1^{\pi_R^{g}}) \derives
        g^{sr\textit{-}\pi^g_S}(\vars{x}_1^{\pi^g_S}), \ \naf\
        g^{c\textit{-}\pi^g_S}(\vars{x}_1^{\pi^g_S})$.
    \end{itemize}

    \noindent Observe that, in this case, disjunctive rules are defined only on the set of relevant
    indices that are not in the key and that each $g^{c\textit{-}\pi^g_S}$ contains only the projection of deleted
    tuples on the set $\pi_S^g$.

\end{enumerate}

\noindent Here, Steps 5--7 handle
relations for which a key is defined and are classified as cycle free. In particular,
if $key(g) \supseteq \pi_R^{g}$ holds, key reparation can be avoided at all (and thus disjunctive rules too);
otherwise
a semi-reparation step is required, but Step 6 identifies further cases in which
even if key reparation is needed, disjunction can be still avoided. Finally, Step 7 handles all the other cases.
Importantly, through Steps 5-7 we take into account only the minimal projections of involved relations
in order to reduce as much as possible computational costs (and even disjunctive rules) not considering
irrelevant attributes.

\begin{enumerate}

  \item[8.] For each node $g \in N_{ncf}$ add the following rules into $\pcqa$:
    \begin{itemize}
      \item $g^c(\vars{x}) \derives g(\vars{x}), \ \naf\
        g_1^{r\textit{-}\pi_R^{d}}(\vars{x}_1^{\pi_R^{d}})$.

      \item[] $g_1^{r\textit{-}\pi_R^{d}}(\vars{x}_1^{\pi_R^{d}}) \derives
        g_1^{r\textit{-}\pi_R^{g}}(\vars{x}_1^{\pi_R^{g}})$.\\
        for each IND $d$ of the form $g(\vars{x}) \rightarrow g_1(\vars{x}_1)$
        such that there is no cycle in $G_q$ involving both $g_1$ and $g$;

      \item $g^c(\vars{x}) \derives g(\vars{x}), \
        \textsf{\#count} \{\vars{x}_{1\exists} : g_1^c(\vars{x}_1)\} =
        \textsf{\#count} \{\vars{x}_{1\exists} : g_1(\vars{x}_1)\}.$\\
        for each IND $d$ of the form
        $\forall \vars{x}_{\forall} \ [ \ g(\vars{x}) \rightarrow \exists \vars{x}_{2\exists} \ g_1(\vars{x}_1) \ ]$
        such that $g_1 \in N_{ncf}$;

      \item $g^{c}(\vars{x}_1) \Or g^{c}(\vars{x}_2) \derives
        g(\vars{x}_1), \ g(\vars{x}_2), \
        \vars{x}_1^{i} \neq \vars{x}_2^{i}.$
        \hfill $\forall i \in \pi$\\
        where $\pi = \{1,\ldots,arity(g)\} - key(g)$ and the key of $g$ is defined by DCs of the form
        $\forall \vars{x}_1, \vars{x}_2 \ \ \neg [g(\vars{x}_1) \wedge g(\vars{x}_2) \wedge
        \vars{x}_1^i \neq \vars{x}_2^i]$;

      \item $g^{r\textit{-}\pi_R^{g}}(\vars{x}^{\pi_R^{g}}) \derives g(\vars{x}), \ \naf\ g^c(\vars{x})$.\\
        if there is at least one node in $N_{cf}$ with an arc
        to $g$, or $g$ appears in $q$;
  \end{itemize}

\item[9.] For each DC of the form $\forall \vars{x}_1, \ldots, \vars{x}_m \ \ \neg [g_1(\vars{x}_1) \wedge \ldots \wedge g_m(\vars{x}_m)
    \wedge \sigma(\vars{x}_1, \ldots, \vars{x}_m)]$ involving
    at least two different relation names (entailing that each $g_i \in N_{ncf}$),
    add the following rules into $\pcqa$:
    \begin{itemize}
      \item $\can{g}_1(\vars{x}_1) \Or \cdots\Or \can{g}_m(\vars{x}_m) \derives g_1(\vars{x}_1), \cdots, g_m(\vars{x}_m),
        \sigma(\vars{x}_1, \ldots, \vars{x}_m).$
    \end{itemize}

\end{enumerate}

\noindent Steps 8 and 9 handle non cycle free relations; the repairing process in this case mimics the standard rewriting,
but projects relations on the relevant attributes whenever possible.

\begin{enumerate}

  \item[10.] For each node $g \in N_{cf}$ if
    $g$ is involved in DCs that do not form a primary key, then add the following
    rules into $\pcqa$:
    \begin{itemize}
      \item $g^{sr}(\vars{x}) \derives
        g(\vars{x}), \ g_1^{r\textit{-}\pi_R^{d_1}}(\vars{x}_1^{\pi_R^{d_1}}), \ \ldots, \
        g_k^{r\textit{-}\pi_R^{d_k}}(\vars{x}_k^{\pi_R^{d_k}})$.

      \item $g_i^{r\textit{-}\pi_R^{d_i}}(\vars{x}_i^{\pi_R^{d_i}}) \derives
        g_i^{r\textit{-}\pi_R^{g_i}}(\vars{x}_i^{\pi_R^{g_i}})$.
        \hfill $\forall i \in$ [1..$k$] s.t. $\pi_R^{g_i} \supset \pi_R^{d_i}$

      \item $g^{c}(\vars{x}_1) \Or \cdots\Or g^{c}(\vars{x}_m) \derives g^{sr}(\vars{x}_1), \cdots, g^{sr}(\vars{x}_m),
        \sigma_d(\vars{x}_1, \ldots, \vars{x}_m).$ \hfill $\forall d$

      \item $g^{r\textit{-}\pi_R^{g}}(\vars{x}^{\pi_R^{g}}) \derives g^{sr}(\vars{x}), \ \naf\ g^c(\vars{x})$.
    \end{itemize}
    where:
    \begin{itemize}
      \item[-] $k \geq 0$ is the number of arcs, labelled by INDs, outgoing from $g$;

      \item[-] atom $g_i^{r\textit{-}\pi_R^{d_i}}$ is in the body of the first rule ($1 \leq i \leq k$)
        iff both $(g,g_i,d_i) \in A$ and $d_i$ is an IND of the form $g(\vars{x}) \rightarrow g_i(\vars{x}_i)$;

      \item[-] $d$ is a DC of the form $\forall \vars{x}_1, \ldots, \vars{x}_m \ \ \neg [g(\vars{x}_1) \wedge \ldots \wedge g(\vars{x}_m)
        \wedge \sigma_d(\vars{x}_1, \ldots, \vars{x}_m)]$
    \end{itemize}
\end{enumerate}
%

\noindent Step 10 handles the special case in which there is no key for a relation but denial constraints are
defined (only) on it.

\begin{enumerate}

  \item[11.]
  For each atom of the form $g(\vars{x})$ in $q$,  replace $g(\vars{x})$ by $g^{r\textit{-}\pi_R^{g}}(\vars{x}^{\pi_R^{g}})$.
\end{enumerate}

\longShort{\begin{example}
Consider again Example \ref{ex-init}; suppose to extend the global schema
by adding the relation $c(code,name)$ which represents the list of customers,
where $code$ is the primary key of $c$.
Moreover, suppose that we ask for the query
    $q(X_c,X_n) \derives c(X_c,X_n), e(X_c,X_n)$
retrieving the customers that are also employees of the bank. In this case, after
building the graph $G_q$ it is easy to see that $m$ is unreachable (so it is discarded)
and that both $c$ and $e$ comply with the requirements described at Steps 5 and 6 of the
optimized algorithm. Consequently, the optimized program under the \emph{CM-complete} semantics is:
\begin{enumerate}
  \item[] $e^{sr\emph{-}1\emph{,}2}(X_c,X_n) \derives e(X_c,X_n)$. \quad $c^{sr\emph{-}1\emph{,}2}(X_c,X_n) \derives c(X_c,X_n)$.

  \item[] $e^{c\emph{-}1\emph{,}2}(X_c,X_n) \derives e^{sr\emph{-}1\emph{,}2}(X_c,X_n), \ e^{sr\emph{-}1\emph{,}2}(X_c,X_n'), \ X_n \neq X_n'$.

  \item[] $c^{c\emph{-}1\emph{,}2}(X_c,X_n) \derives c^{sr\emph{-}1\emph{,}2}(X_c,X_n), \ c^{sr\emph{-}1\emph{,}2}(X_c,X_n'), \ X_n \neq X_n'$.

  \item[] $e^{r\emph{-}1\emph{,}2}(X_c,X_n) \derives e^{sr\emph{-}1\emph{,}2}(X_c,X_n), \ not \ e^{c\emph{-}1\emph{,}2}(X_c,X_n)$.

  \item[] $c^{r\emph{-}1\emph{,}2}(X_c,X_n) \derives c^{sr\emph{-}1\emph{,}2}(X_c,X_n), \ not \ c^{c\emph{-}1\emph{,}2}(X_c,X_n)$.

  \item[] $\qcqa(X_c,X_n) \derives c^{r\emph{-}1\emph{,}2}(X_c,X_n), \ e^{r\emph{-}1\emph{,}2}(X_c,X_n)$.
\end{enumerate}
Note that, since both $e$ and $c$ are not affected by IND violations, and they have no irrelevant variables,
the semi-reparation step cannot actually discard tuples.
However, the obtained program is non-disjunctive and stratified. Thus, it can be evaluated in polynomial time \cite{leon-etal-2002-dlv}.

In this case, the only answer set of the program contains the consistent answers to the original query. \qed
\end{example}}
{}

\paragraph{\emph{\textbf{$\Sigma = \emph{loosely-exact}$}.}}
In Section \ref{sub:generalencoding} we proved that there are common cases in which CQA under the
\emph{loosely-exact} semantics and the \emph{CM-complete} semantics actually coincide.
As a consequence, in these cases, all the optimizations defined for the \emph{CM-complete} semantics apply
also to the \emph{loosely-exact} semantics.


\section{Experiments}
\label{sec:experiments}

In this section we present some of the experiments we carried out to assess the effectiveness of
our approach to consistent query answering.

Testing has been performed by exploiting our complete system for data integration,
which is intended to simplify both the integration system design
and the querying activities by exploiting a user-friendly GUI.
Indeed, this system  both supports the user in designing the global schema and the mappings
between global relations and source schemas, and
it allows to specify user queries over the global schema via a QBE-like interface.
The query evaluation engine adopted for the tests is \dlvdb
\cite{TerracinaDeFrancescoPanettaLeone08}
coupled, via ODBC, with a PostgreSQL DBMS where input data were stored.
\dlvdb is a DLP evaluator born as a database oriented extension of the well known  \dlv system
\cite{leon-etal-2002-dlv}. It has been recently extended for dealing with
unstratified negation, disjunction and external function calls.

We first address tests on a real world scenario and then report on tests for scalability issues on synthetic
data.

\subsection{Tests on a real world scenario}
\paragraph{Data Set.} We have exploited the real-world data integration framework developed in the INFOMIX project
(IST-2001-33570) \cite{LeoneGrecoIanni05} which integrates data from a real university context. In
particular, considered data sources were available at the University of Rome ``La Sapienza''.
These comprise information on students, professors, curricula and exams in various faculties of
the university.

There are about 35 data sources in the application scenario, which are mapped into 12 global
schema relations with  20 GAV mappings and 21 integrity constraints. We call this data set
{\bf Infomix} in the following.
\longShort{Figure \ref{fig:infomixgraph} reproduces the main characteristics of
the global database:
each node corresponds to a global relation showing its arity and key. An edge between $r_1$ and $r_2$
labelled by $r_1[I] \subseteq r_2[J]$ indicates an IND of the form \indNorm{r} where
$I$ and $J$ are the positions of $\vars{x}_{\forall}$ in $\vars{x}_1$ and $\vars{x}_2$, respectively;
the arc is labelled with the attributes of $a$ and $b$
involved in the IND. Observe that there are cyclic INDs involving \emph{teaching}, \emph{exam\_record}
and \emph{professor}.
\begin{figure}[t]
\centering
\includegraphics[width=\textwidth]{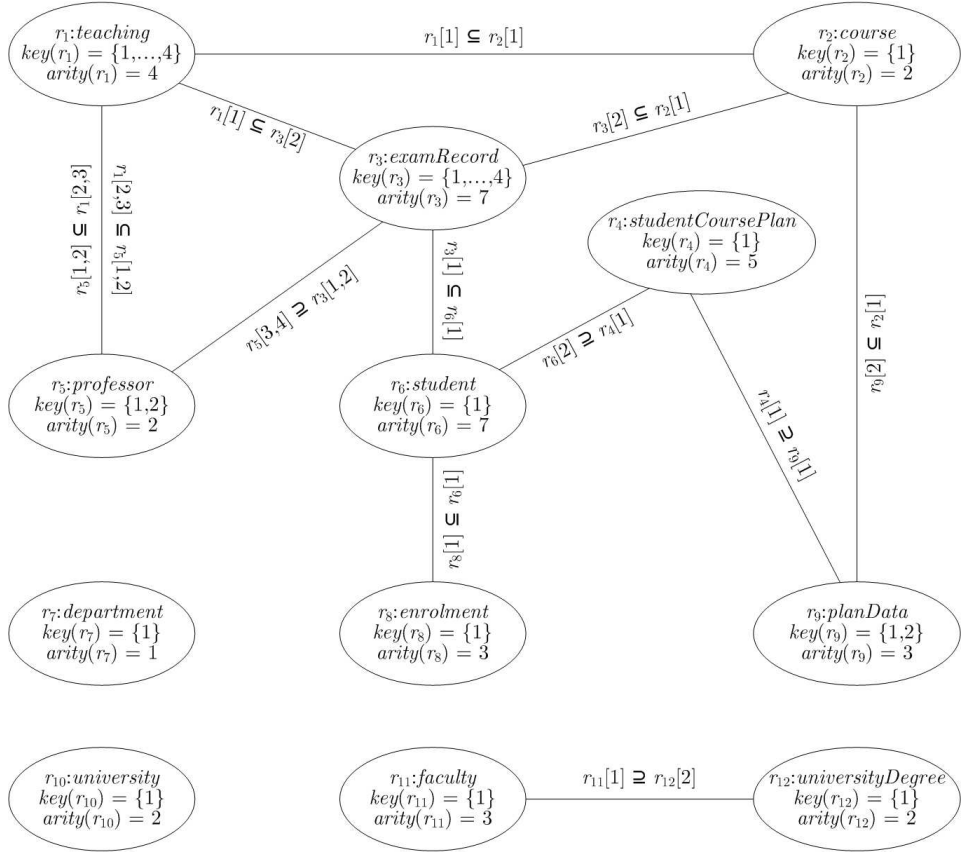}
\caption{INFOMIX database.}\label{fig:infomixgraph}
\end{figure}
}
{(A detailed description of the {\bf Infomix} framework is given in
the extended version \cite{tplpOnCoRR} of this paper.)}

Besides the original source database instance (which takes about 16Mb on DBMS), we obtained
bigger instances artificially. Specifically, we generated a number of copies of the original
database; each copy is disjoint from the other ones but maintains the same data correlations
between instances as the original database. This has been carried out by mapping each original
attribute value to a new value having a copy-specific prefix.

Then, we considered two further datasets, namely {\bf Infomix-x-10} and {\bf Infomix-x-50}
storing 10 copies (for a total amount of 160Mb of data) and 50 copies (800Mb) of the original
database, respectively. It holds that {\bf Infomix} $\subset$ {\bf Infomix-x-10} $\subset$ {\bf Infomix-x-50}.

\paragraph{Compared Methods and Tested Queries.}
In order to assess the characteristics of the proposed optimizations, we measured the execution
time of different queries with {\em (i)} the standard encoding (identified as {\bf STD} in the following),
{\em (ii)} a na\"{\i}ve optimization obtained by only removing relations not strictly needed for answering
the queries ({\bf OPT1} in the following), and
{\em (iii)} the fully optimized encoding presented in Section \ref{sec:cqa} ({\bf OPT2} in the following).
Each of these cases has been evaluated for the
three semantics considered in this paper.
In order to isolate the impact of our optimizations, we
disabled other optimizations (like magic sets) embedded in the datalog evaluation engine. Clearly,
such optimizations are complementary to our own and might further improve the overall performances.

\newcommand{\tabs }{\noindent\hspace{0.5cm}}
\newcommand{\tabl }{\noindent\hspace{1.0cm}}

\longShort{

Tested queries are as follows:
\begin{alltt} \footnotesize
\tabs Q1(X1) :- course(X2,X1), plan\_data(PL,X2,\_),

    \tabl student\_course\_plan(PL,"09089903",\_,\_,\_).

\tabs Q2(X1) :- university(X1,\_).

\tabs Q3(X1,X2,X3) :- university\_degree(X1,X2), faculty(X2,\_,X3).

\tabs Q4(X1,X2,X3) :- student(S,\_,X1,\_,\_,\_,\_), enrollment(S,\_,\_),

    \tabl exam\_record(S,\_,\_,X2,X3,\_,\_), S == "09089903".

\tabs Q5(X1,X2) :- student\_r(S1,\_,X1,\_,\_,\_,\_), exam\_record\_r(S1,C,\_,\_,\_,\_,\_),

    \tabl student\_r(S2,\_,X2,\_,\_,\_,\_), exam\_record\_r(S2,C,\_,\_,\_,\_,\_),

    \tabl S1 == "09089470", S1<>S2.

\tabs Q6(X1,X2,X3) :- student(X1,\_,\_,\_,\_,\_,\_), exam\_record(X1,\_,\_,X2,X3,\_,\_),

    \tabl X1 == "09089903".
\end{alltt}\normalsize

Observe that {\tt Q2} involves key constraints only, {\tt Q1}, and {\tt Q3} involve both
keys and acyclic INDs; specifically, {\tt Q3} involves a SFK while {\tt Q1} involves NKC INDs.
Finally, {\tt Q4}, {\tt Q5} and {\tt Q6} involve keys and cyclic NKC INDs.}
{We have considered six queries, fully specified in the
the extended version \cite{tplpOnCoRR} of this paper, named {\tt Q1} $\cdots$ {\tt Q6}.
In particular, {\tt Q2} involves key constraints only, {\tt Q1}, and {\tt Q3} involve both
keys and acyclic INDs; specifically, {\tt Q3} involves a SFK while {\tt Q1} involves NKC INDs.
Finally, {\tt Q4}, {\tt Q5} and {\tt Q6} involve keys and cyclic NKC INDs.}

\paragraph{Results and discussion.}

\begin{figure}[tbp]
\centering
\begin{tabular}{cc}
\includegraphics[width=5.5cm]{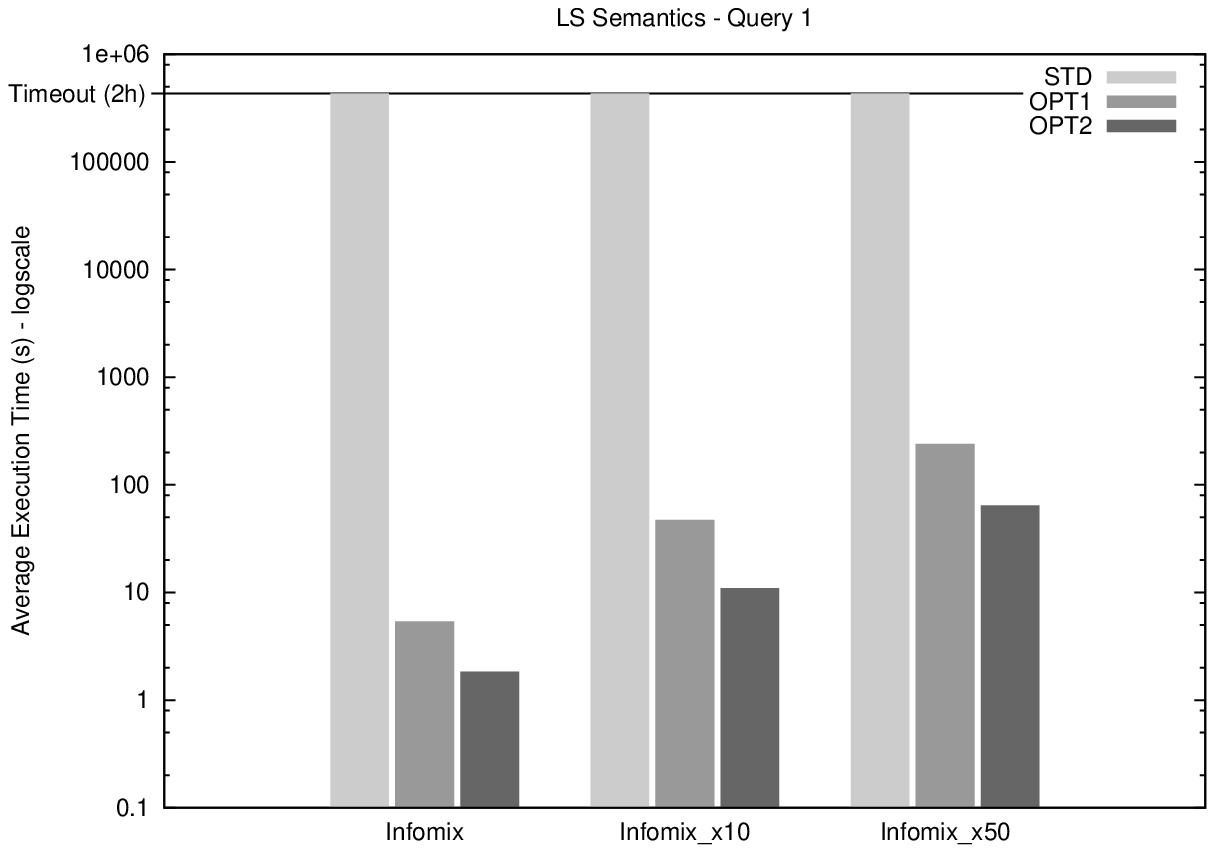} &
\includegraphics[width=5.5cm]{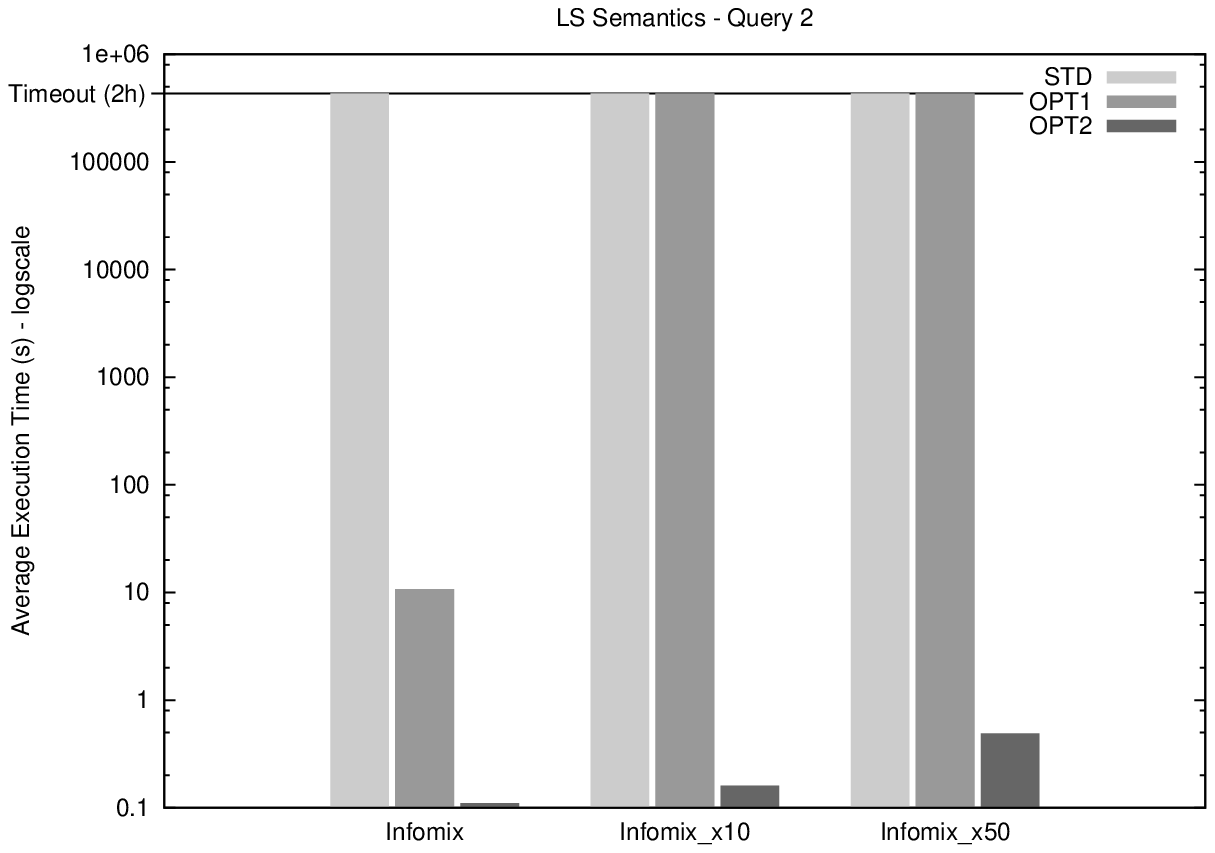} \\
\includegraphics[width=5.5cm]{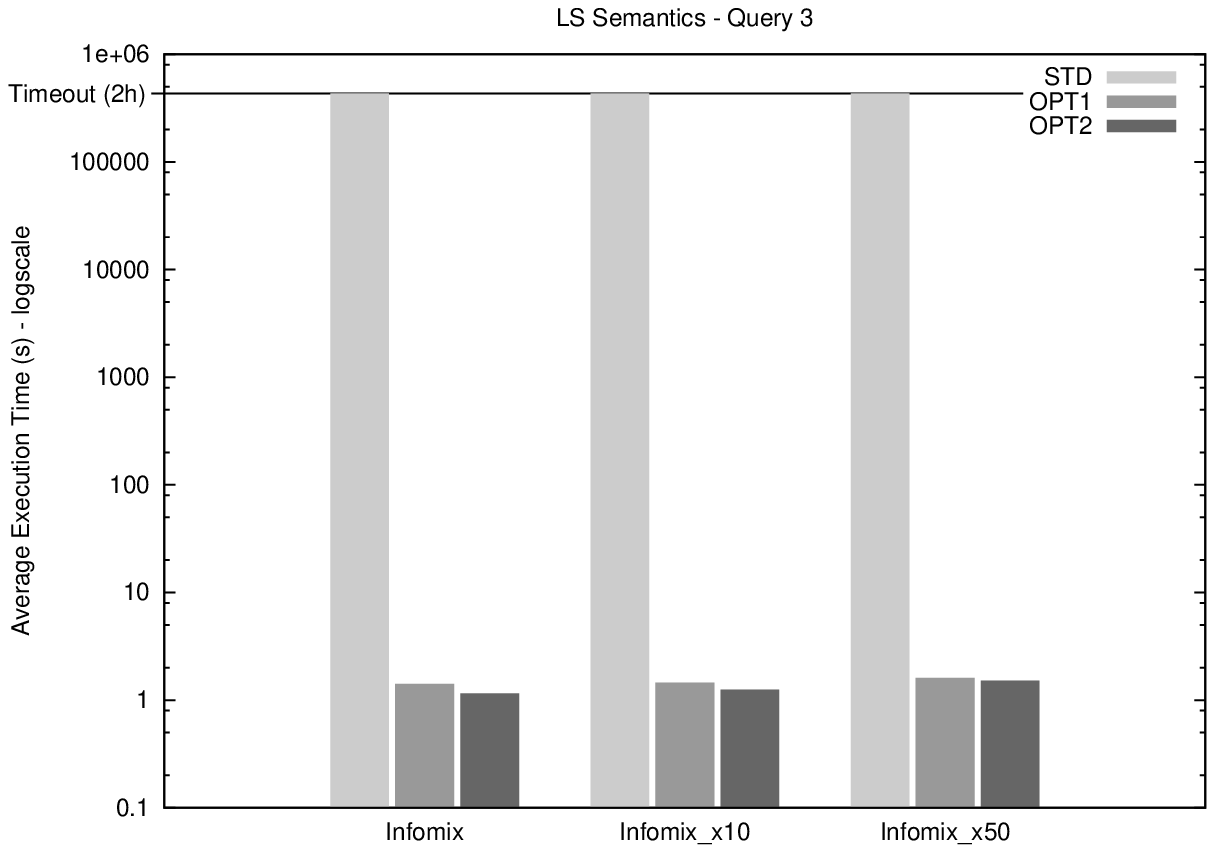} &
\includegraphics[width=5.5cm]{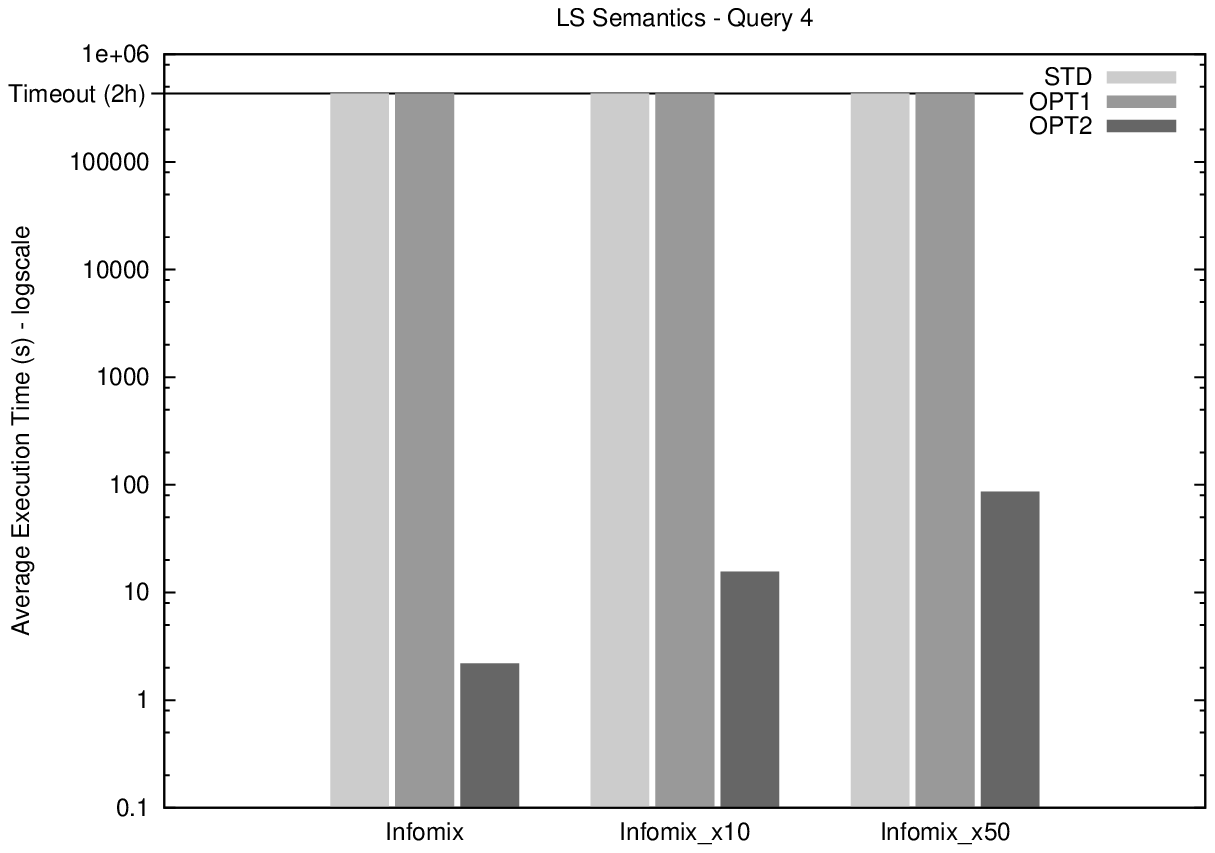} \\
\includegraphics[width=5.5cm]{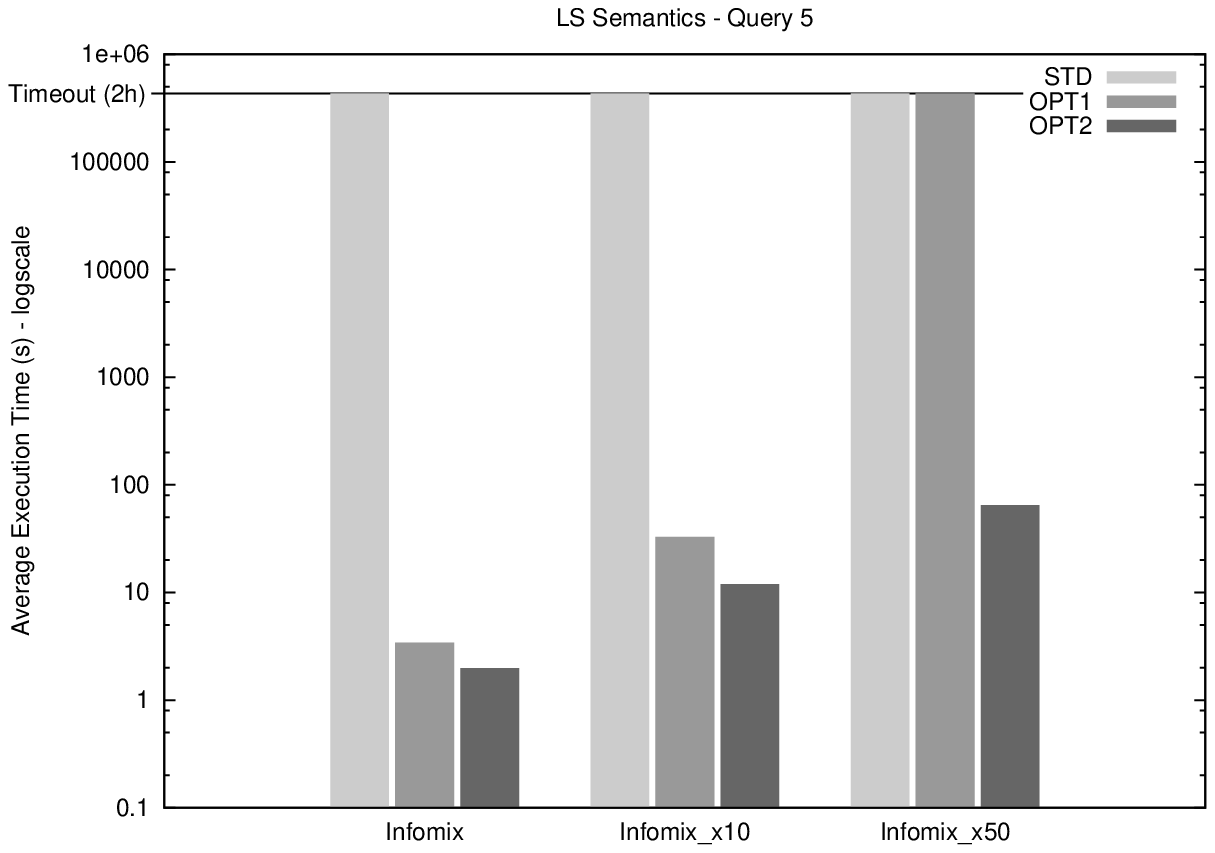} &
\includegraphics[width=5.5cm]{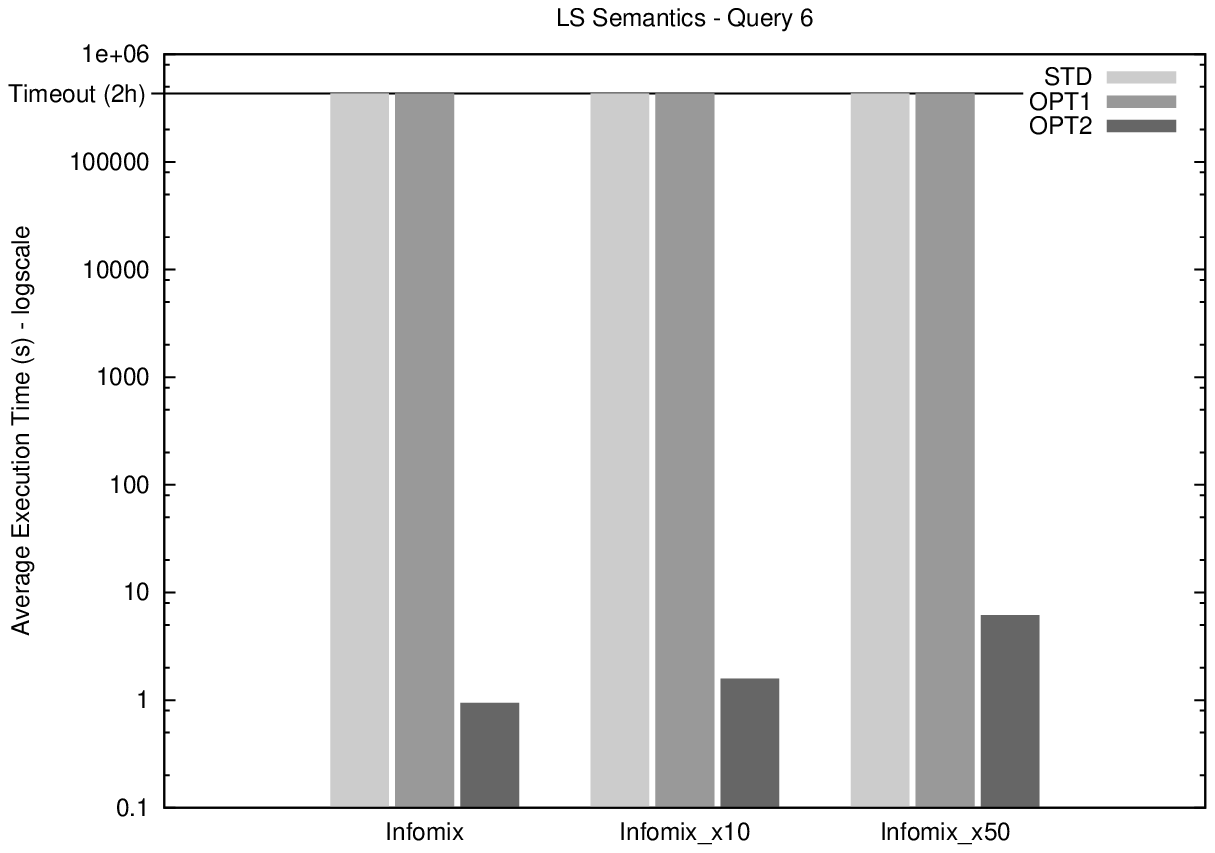}
\end{tabular}
\caption{Query evaluation execution times for the \emph{loosely-sound} semantics.}\label{fig:result-ls} 
\end{figure}

\begin{figure}[htbp]
\centering
\begin{tabular}{cc}
\includegraphics[width=5.5cm]{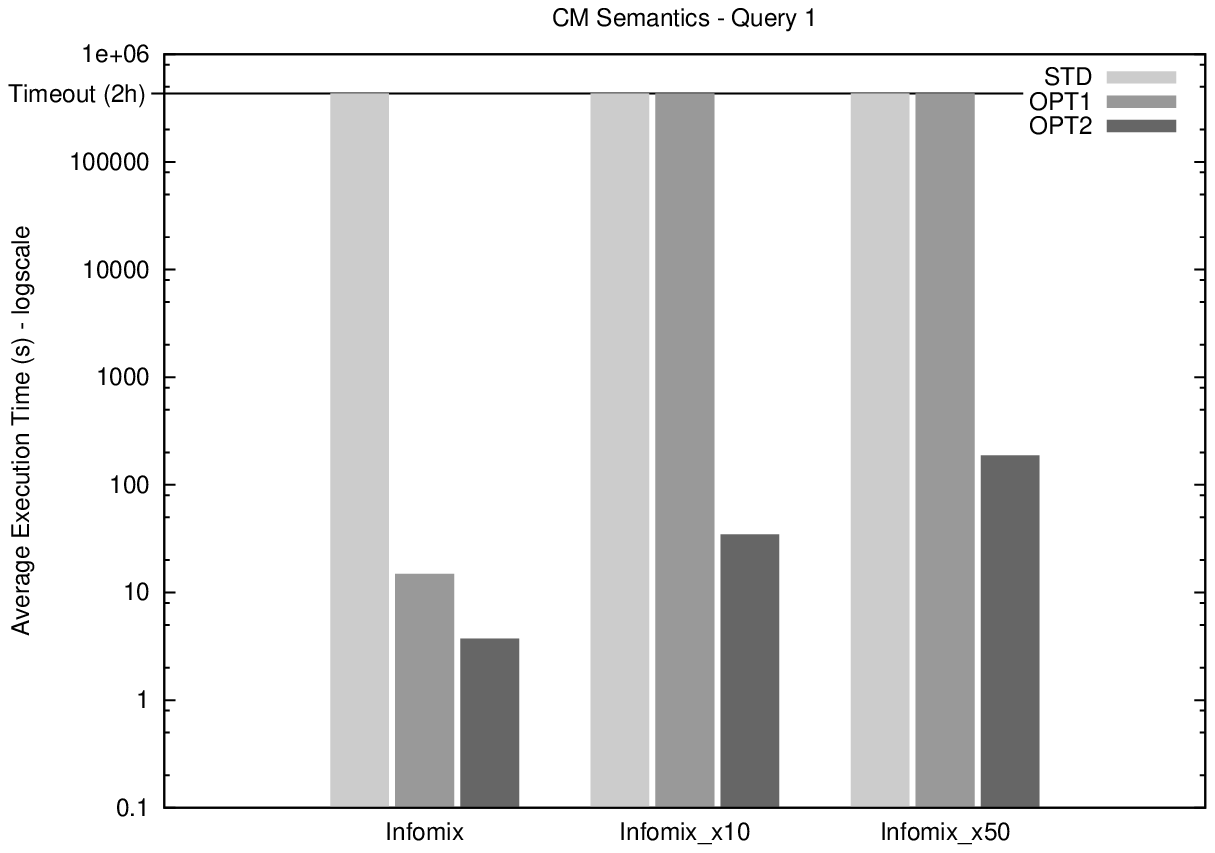} &
\includegraphics[width=5.5cm]{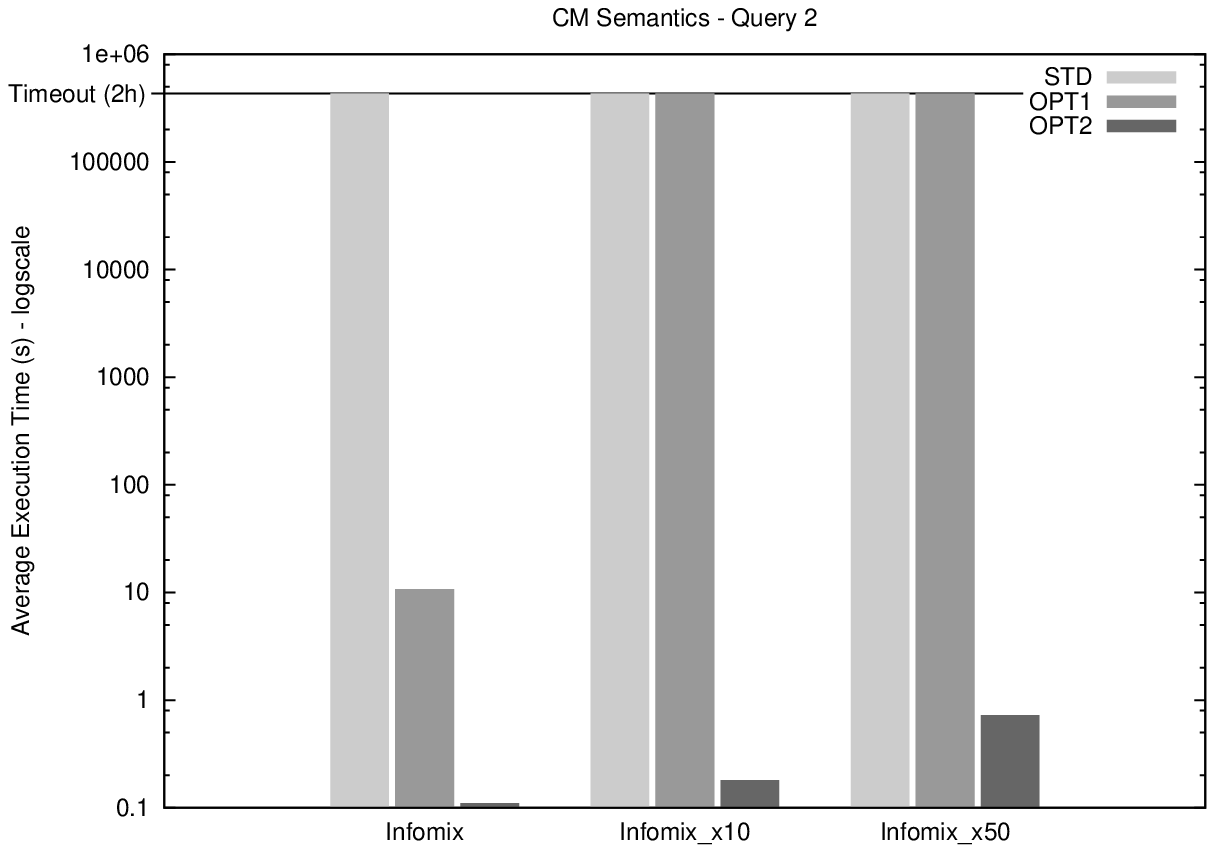} \\
\includegraphics[width=5.5cm]{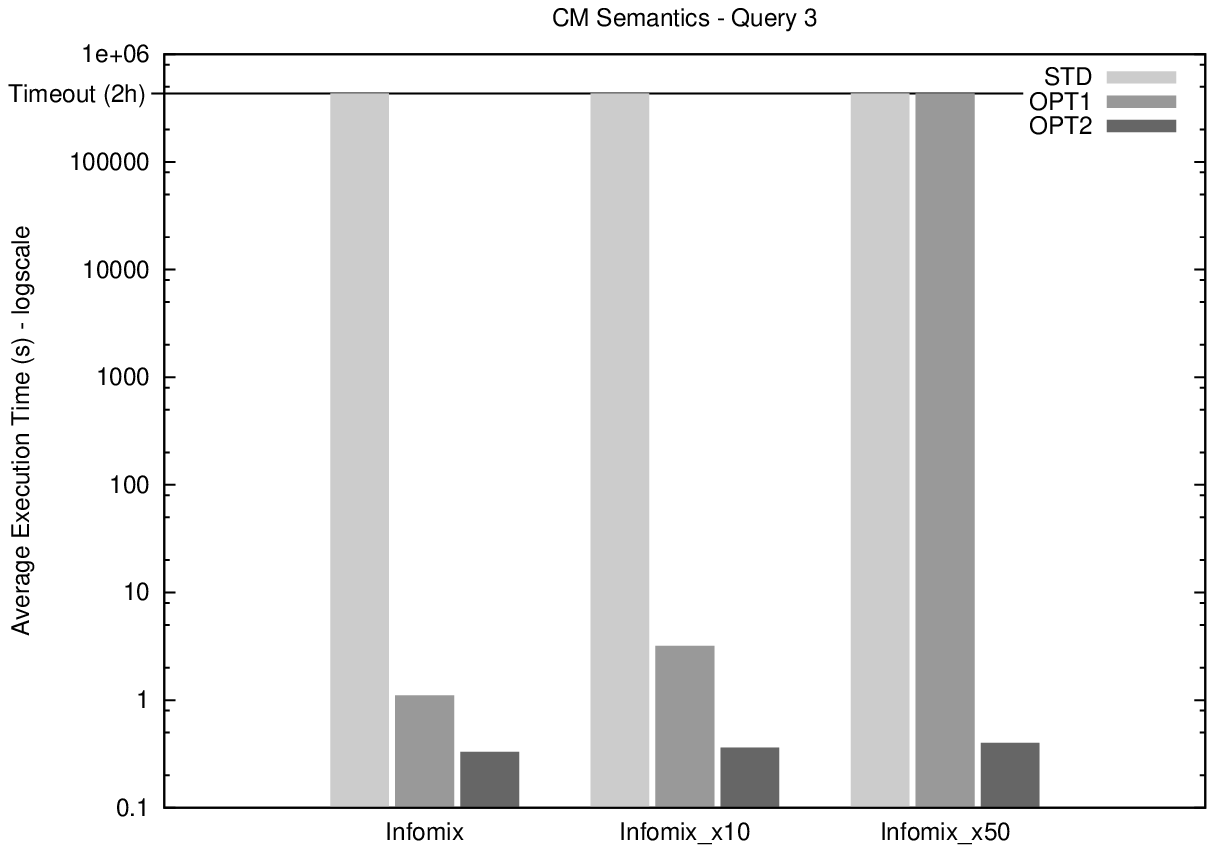} &
\includegraphics[width=5.5cm]{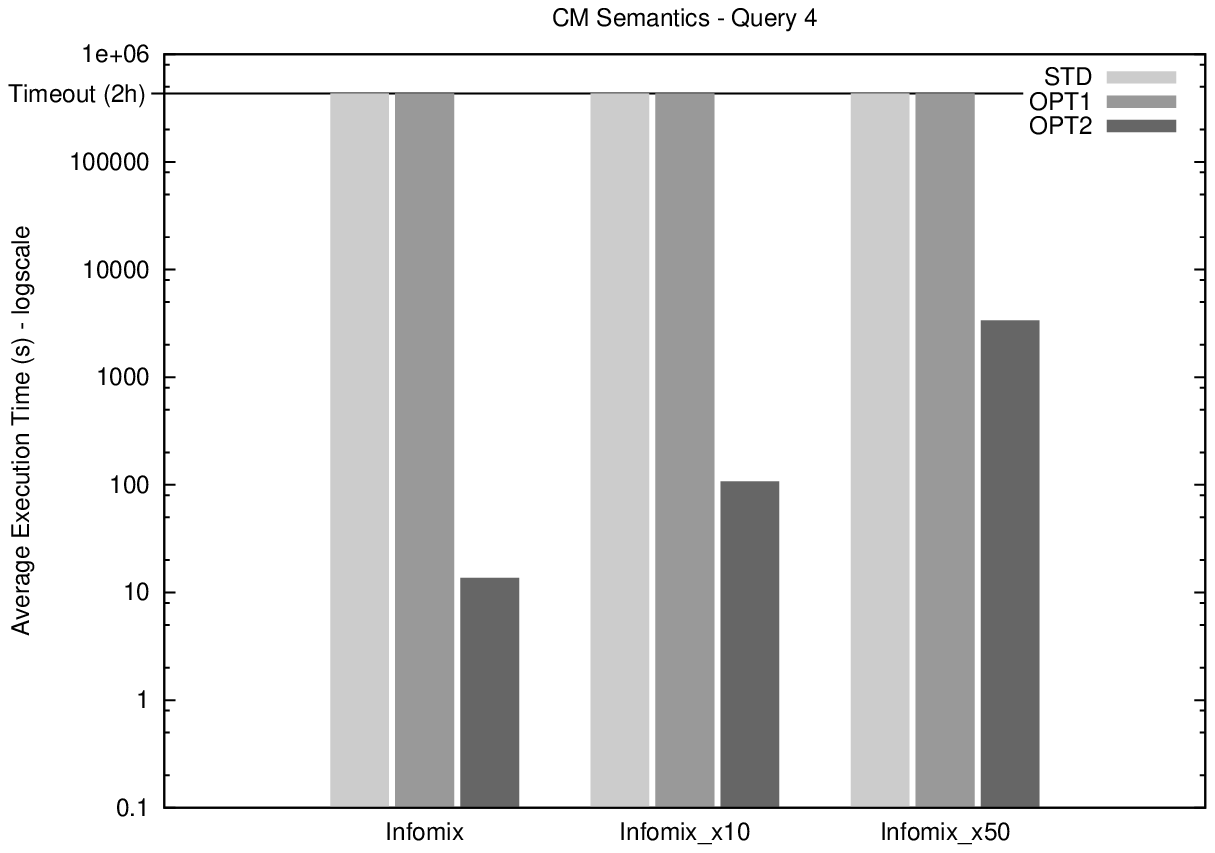} \\
\includegraphics[width=5.5cm]{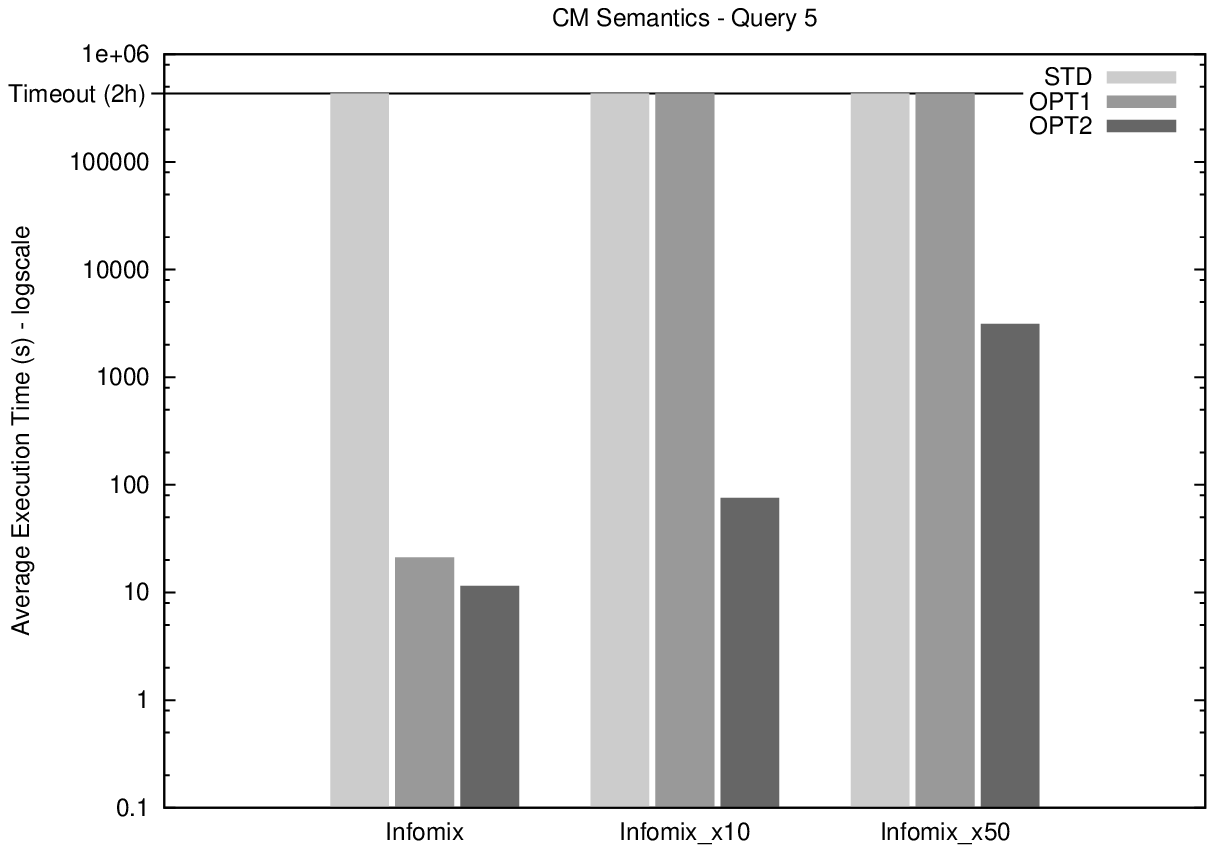} &
\includegraphics[width=5.5cm]{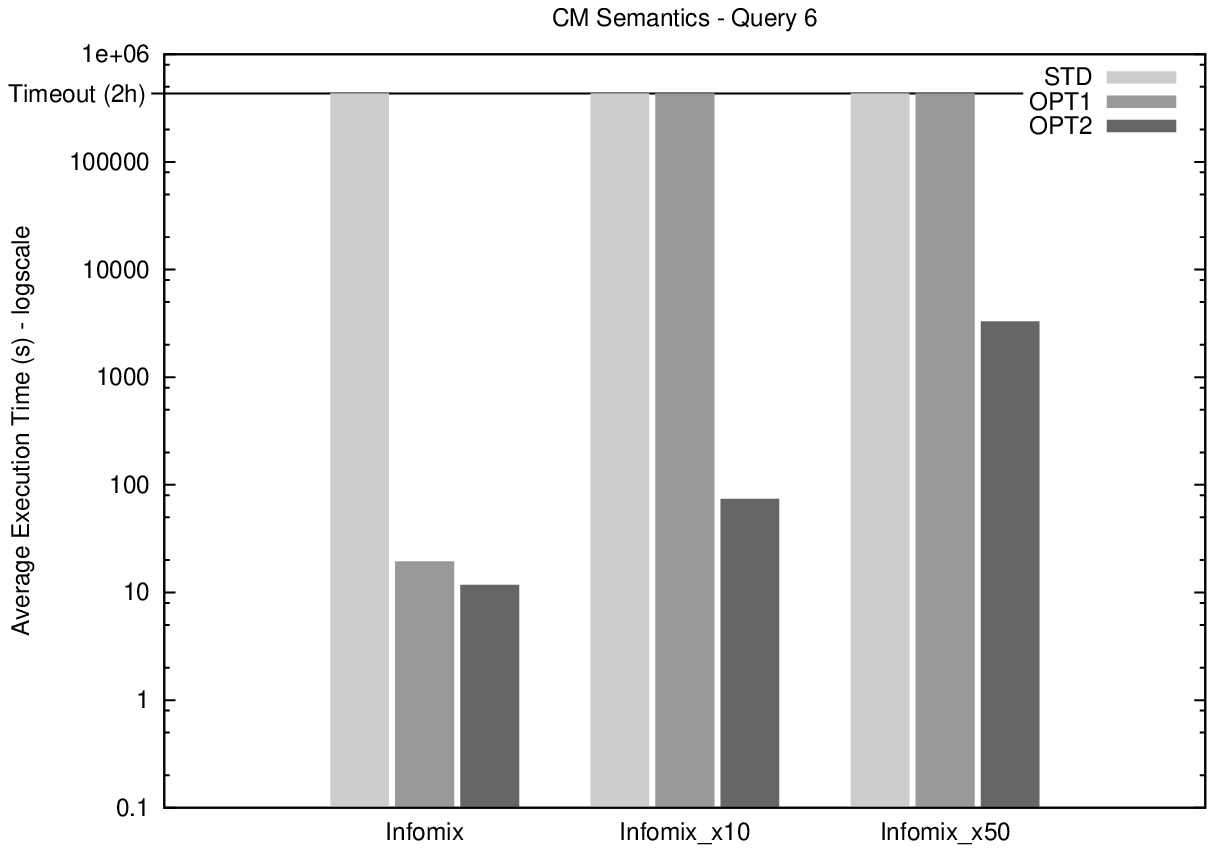}
\end{tabular}
\caption{Query evaluation execution times for the \emph{CM-Complete} semantics.}\label{fig:result-cm} 
\end{figure}

All tests have been carried out on an Intel Xeon X3430, 2.4 GHz, with 4 Gb Ram, running
Linux Operating System. We set a time limit of 120 minutes after which query execution has
been killed. Figures \ref{fig:result-ls} and \ref{fig:result-cm} show obtained results for the
\emph{loosely-sound} and the \emph{CM-complete} semantics.
It is worth recalling that, as we pointed out in Section \ref{sub:l-e}, optimizations for the \emph{loosely-exact}
semantics are inherent to the equivalence classes to the \emph{CM-complete} semantics discovered in this paper.
As a consequence,
we tested this semantics only on queries Q2 and Q3 for which such equivalence holds.
Then, since the execution times of the optimized encoding coincide with the \emph{CM-complete} graphs
for queries Q2 and Q3, we do not report specific figures for them.

Analyzing the figures, we observe that:
the proposed optimizations do not introduce computational overhead and, in most cases, transform
practically untractable queries in tractable ones; in fact, for all the tested queries the execution
time of the standard rewriting exceeded the time limit.
{\bf OPT1} helps mostly on the smallest data set; in fact for {\bf Infomix-x-10} it shows some gain in 33\% of cases
and only in two cases for {\bf Infomix-x-50}.

As for the comparison among the optimized encodings, we can observe that if INDs are not involved by the query
({\tt Q2}) the \emph{loosely-sound} and the \emph{CM-complete} optimizations
have the same performances; this confirms theoretical expectations.
When acyclic INDs are involved ({\tt Q1}, {\tt Q3}),
the \emph{loosely-sound} optimization
performs slightly better because the \emph{CM-complete} must choose
the tuples to be deleted due to IND violations, whereas the \emph{loosely-sound} semantics just works on the original data.
Finally, when involved INDs are cyclic ({\tt Q4}, {\tt Q5}, {\tt Q6}) the performance of the
\emph{CM-complete} optimization
further degrades w.r.t. the \emph{loosely-sound} one because recursive aggregates must be exploited to choose
deletions and, this,
increases the complexity of query evaluation.

\subsection{Scalability analysis w.r.t. the number and kind of constraint violations}
Since, in the real world scenario emerged that the \emph{CM-complete} semantics is more affected
than the \emph{loosely sound} one from the kind
of involved constraints, we carried out a scalability analysis on this semantics,
whose results are reported  next.

We considered a synthetic data set composed of three relations named $r_1, r_2$, and $r_3$ over which
we imposed different sets of ICs in order to analyze the scalability of our methods
depending on the presence of keys and/or in presence/absence of acyclic and cyclic INDs.
In particular, we imposed the following key constraints: $key(r_2)=\{1,2\}$, $key(r_3)=\{1\}$,
and we experimented with three different sets of INDs:
$NOINCL = \emptyset$, $ACYCLIC=\{r_1(X_1,X_2,X_3,X_4)\rightarrow r_2(X_2,X_5,X_3,X_6)$,
$r_1(X_1,X_2,X_3,X_4)\rightarrow r_3(X_1,X_5,X_6,X_7)\}$ and
$CYCLIC = ACYCLIC \cup \{r_2(X_1,X_2,X_3,X_4) \rightarrow r_1(X_5,X_6,X_7,X_2)\}$. The employed query is:
$query(X1,X3) \derives r_1(X1,X2,X3,X4), r_2(X2,X3,X5,X6)?$
We have randomly generated synthetic databases having a growing number of key violations on table $r_2$.
The generation process progressively adds key violations to $r_2$ by generating pairs of conflicting tuples;
after an instance of $r_2$ is obtained, tables $r_1$ and $r_3$ are generated by taking
values from $r_2$ in such a way that INDs are satisfied.
In addition, for each tuple of $r_3$ a key-conflicting tuple is generated.
In order to assess the impact of the number of INDs violations, for each database instance
$DB_x$, containing $x$ key violations on table $r_2$, we generated a $DB_{x}$-$10$ instance where
the $10\%$ of tuples is (randomly) removed from tables $r_1$ and $r_3$ (causing INDs violations). 
We have generated six database instances per size (number of key violations on table $r_2$),
and plotted the time (averaged over the instances of the same size) in Figure~\ref{fig:scalability}.

In detail, Figure \ref{fig:scalability}(a) shows the results for incrementally higher KD violations
with no IND violations. Both standard and optimized encodings have been tested. Figure \ref{fig:scalability}(b)
compares the optimized encoding only, when the percentage of IND violations is 0\% or 10\%.
Observe that, in general, even when there is no initial IND violation, the KD repairing process may induce some of them.

The analysis of these figures shows that even if cyclic INDs are generally harder, their scaling is
almost the same as the acyclic ones. On the contrary, in the absence of INDs the optimization may boost the
performances (see the flat line in Figure \ref{fig:scalability}(a)). Figure \ref{fig:scalability}(b)
points out that when the number of IND violations increases, the performance may improve. This
behavior is justified by the fact that tuple deletions due to IND repairs may, in their turn, remove
KD violations. This reduces the number of disjunctions to be evaluated.

\begin{figure}[tbp]
\centering
\begin{tabular}{cc}
\includegraphics[width=6.2cm]{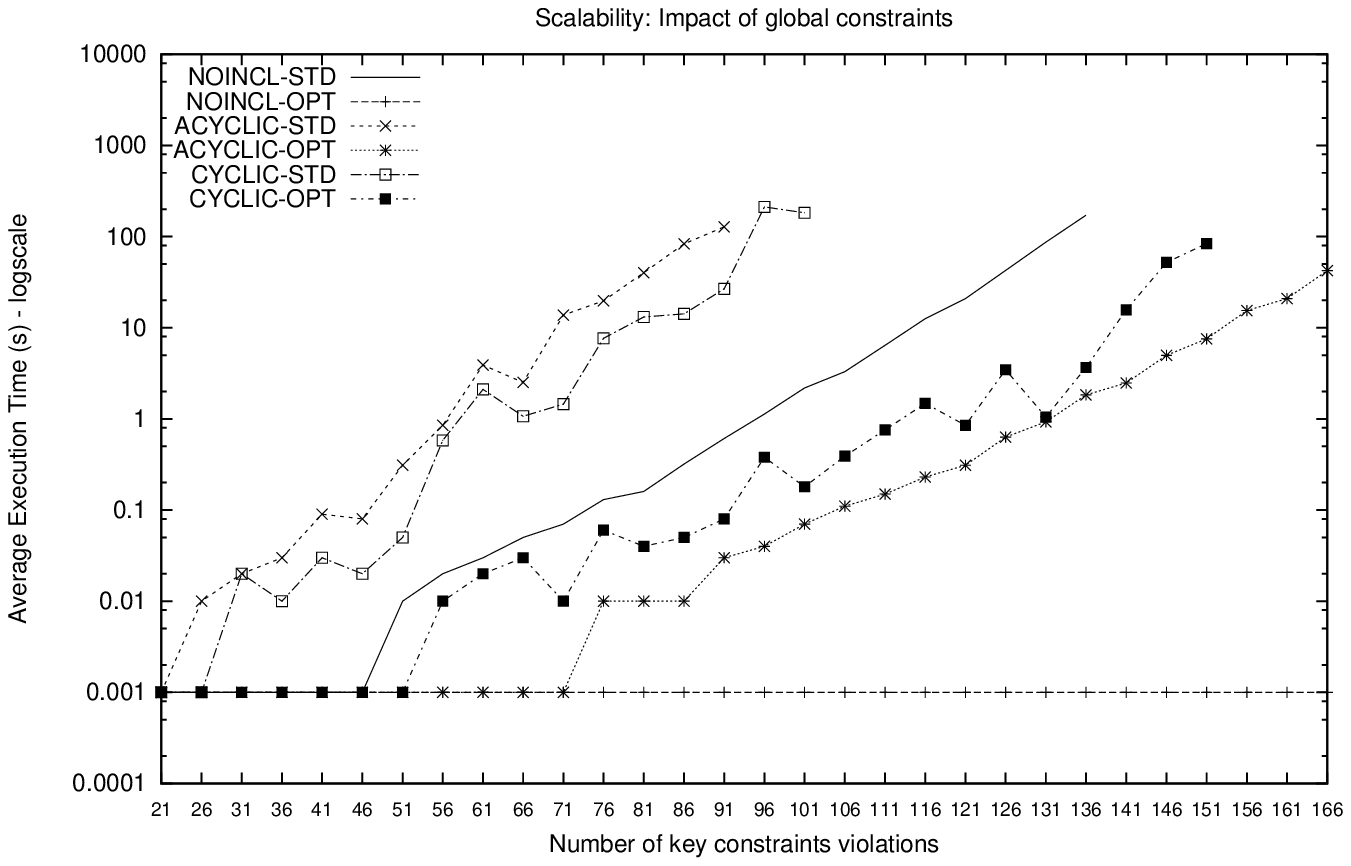} &
\includegraphics[width=6.2cm]{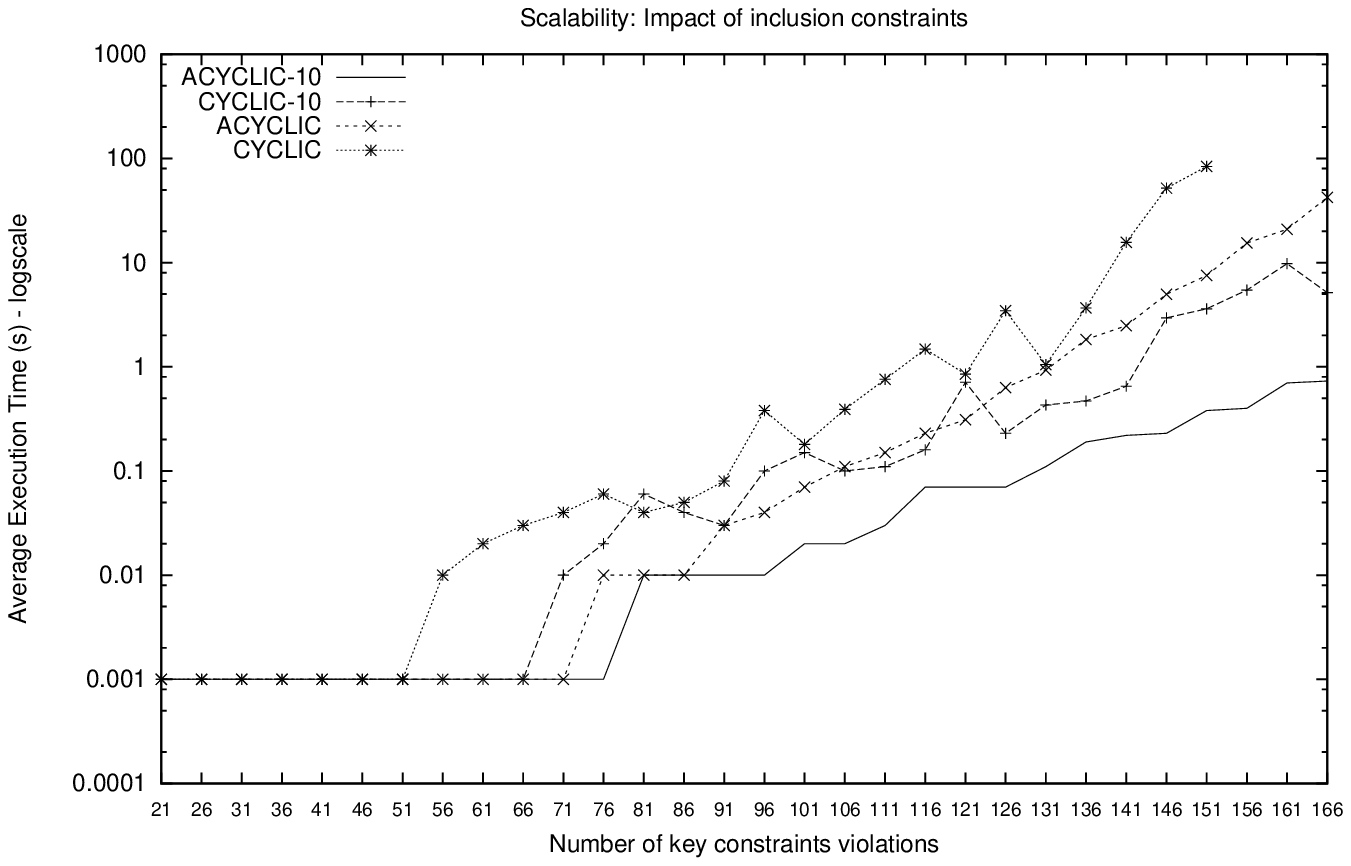} \\
(a) & (b)
\end{tabular}
\caption{Scalability Analysis}\label{fig:scalability} 
\end{figure}

\section{Related work and concluding remarks} \label{sec:conlusions}


From the 90ies -- when the founding notions of \emph{CQA}~\cite{Bry97},
\emph{GAV mapping}~\cite{Garcia-MolinaPapakonstantinou97,TomasicRaschidValduriez98,GohBressanMadnick99},
and \emph{database-repair}~\cite{ArenasBertossiChomicki99} were introduced --
\emph{data integration}~\cite{Lenzerini02} and \emph{inconsistent databases}~\cite{BertossiHunterSchaub05} 
have been studied quite in depth. 

Detailed characterizations of the main problems arising in a data integration system have been provided,
taking into account different semantics, constraints, and query types
\cite{CaliLemboRosatiPODS03,CaliLemboRosati03,ArenasBertossiChomicki03,ChomickiMarcinkowski05,GriecoLemboRosatiRuzzi05,FuxmanMiller07,EiterFinkGrecoLembo08}.



This paper provides a contribution in this scenario by extending the decidability boundaries for the \emph{loosely-exact} semantics
(as called in \citeNP{CaliLemboRosatiPODS03} but firstly introduced by \citeNP{ArenasBertossiChomicki99})
and the \emph{loosely-sound} semantics, in case of both KDs and SFSK INDs.
%


A first proposal of an unifying framework for CQA in a Data Integration setting
is presented in \cite{CaliLemboRosati05} using first-order logic; it considers different semantics
defined by interpreting the mapping assertions between the global and the local schemas of the data integration system.
A common framework for computing repairs in a single database setting is proposed in \cite{EiterFinkGrecoLembo08};
it covers a wide range of semantics relying on the general notion of preorder for candidate repairs, but only
universally quantified constraints are allowed. Moreover, the authors introduce an abstract logic programming framework
to compute consistent answers. Finally, the authors propose an optimization strategy called factorization
that, as will be clarified below, is orthogonal to our own.

This paper provides a contribution in this setting since it unifies different semantics, as in
\cite{CaliLemboRosati05} and \cite{EiterFinkGrecoLembo08}, but also provides an algorithm that, given a retrieved database,
a user query $q$, and a semantics, automatically composes an ASP program capable of computing the
consistent answers to $q$.
In particular, our ASP-rewriting offers a natural, compact, and direct way
for encoding even hard cases where the CQA problem belongs to the \PiP{2} complexity class.

Theoretical studies gave rise to concrete implementations
most of which were conceived to operate on some specific semantics and/or constraint types.
\cite{ArenasBertossiChomicki99,CaliCalvaneseDeGiacomoLenzerini02,GrecoZumpano00,GrecoGrecoZumpano01,CaliLemboRosati03,ArenasBertossiChomicki03,ChomickiMarcinkowskiStaworko04,CaliCalvaneseDeGiacomoLenzerini04,ChomickiMarcinkowskiStaworkoHippo04,LemboPhDthesis04,GriecoLemboRosatiRuzzi05,LeoneGrecoIanni05,FuxmanFazliMiller05,FuxmanMiller07}.
As an example, in \cite{LeoneGrecoIanni05} only the \emph{loosely-sound} semantics was supported. In this paper,
we provide both a unified framework based on ASP, and a complete system supporting (i) all the three aforementioned significant semantics 
in case of conjunctive queries and the most commonly used database constraints (KDs and INDs), (ii)
specialized optimizations, and (iii) a user-friendly GUI.

Another general contribution of our work comes from a novel optimization technique
that, after analyzing the query and localizing a minimal number of relevant ICs, tries to ``simplify''
their structure to reduce the number of database repairs -- as they could be exponentially many \cite{ArenasBertossiChomicki01}.
Such technique could be classified as ``vertical'' due to the fact that it reduces (whenever possible)
the arity
of each active relation (with the effect, e.g., of decreasing the number of key conflicts)
without looking at the data.
It is orthogonal to other ``horizontal'' approaches,
such as magic-sets \cite{FaberGrecoLeone07} and factorization \cite{EiterFinkGrecoLembo08}
which are based on data filtering strategies.
In particular, a system exploiting ASP incorporating magic-set techniques for CQA is described in \cite{caniup-bertossi-2010}.
Other approaches complementary to our own are based on first-order rewritings of the query
\cite{ArenasBertossiChomicki99,ChomickiMarcinkowski02,CaliLemboRosati03,GriecoLemboRosatiRuzzi05,FuxmanMiller07}.

The combination of our optimizations with such approaches, and further extensions of decidability boundaries
for CQA are some of our future line of research.

\paragraph{\textbf{\emph{Acknowledgments}}.}
This work has been partially supported by the Calabrian Region under PIA (Pacchetti Integrati
di Agevolazione industria, artigianato e servizi) project DLVSYSTEM approved in BURC n.
20 parte III del 15/05/2009 - DR n. 7373 del 06/05/2009.

\bibliographystyle{acmtrans}
\bibliography{mannaGlobal}

\end{document}